\documentclass[11pt]{article}
\usepackage{mymacros}
\usepackage{authblk}
\usepackage{setspace}
\usepackage[title]{appendix}
\usepackage{arydshln}
\newcommand{\thmax}{\widetilde{H}_{\max}}
\newcommand{\phmax}{{H^{'}_{\max}}}
\newcommand{\good}{\textsc{good}}
\newcommand{\wt}{\widetilde{\Theta}}
\newcommand{\wtt}[1]{\widetilde{#1}}
\newcommand{\wdt}[1]{\accentset{\approx}{#1}}

\newcommand{\nice}{\textsc{nice}}
\newcommand{\opt}{\textsc{opt}}
\newcommand{\name}{ \texttt{FewQubits} }
\newcommand{\kd}{\texttt{KD\_OneShot}}
\newcommand{\mix}{\textup{mix}}
\newcommand{\alice}{\textup{alice}}
\newcommand{\bob}{\textup{bob}}
\newcommand{\ip}[2]{\left\langle#1, #2 \right\rangle}
\newcommand{\pure}{\textup{pure}}
\newtheorem{assumption}[theorem]{Assumption}

\title{One-Shot Non-Catalytic Distributed Purity Distillation\footnote{A preliminary version of this work appeared at the 2023 59th Annual Allerton Conference on Communication, Control, and Computing \cite{Allerton}}\hspace{2mm}\footnote{All authors are listed in alphabetical order.}}
\author[1]{Francesco Buscemi\thanks{buscemi@i.nagoya-u.ac.jp}}
\author[2]{Sayantan Chakraborty\thanks{kingsbandz@gmail.com/schakra4@uni-koeln.de}}
\author[3,4]{Rahul Jain\thanks{rahul@comp.nus.edu.sg}}
\author[5]{Aditya Nema\thanks{aditya.nema30@gmail.com}}
\author[6]{Pranab Sen\thanks{pranab.sen.73@gmail.com/ psen@tifr.res.in}}
\affil[1]{Department of Mathematical Informatics, Nagoya University}
\affil[2]{Quantum Information and Computation, Abteiling f\"ur Informatik, Universit\"at zu K\"oln}
\affil[3]{Centre for Quantum Technologies, National University of Singapore}
\affil[4]{Department of Computer Science, National University of Singapore and MajuLab, UMI 3654, Singapore}
\affil[5]{Department of Electrical Engineering, IIT Delhi}
\affil[6]{School of Technology and Computer Science, Tata Institute of Fundamental Research, Mumbai}

  \date{}

\begin{document}
\maketitle

\begin{abstract}
    
Pure states are an important resource in many quantum information
processing protocols. However, even making a fixed pure state, say
$\ket{0}$, in the laboratory requires a considerable amount of effort.
Often one ends up with a mixed state $\rho$ whose classical
description is nevertheless known. Hence it is important to develop
protocols that extract a fixed pure state from a known mixed state. In this work, we study the problem of extracting a fixed pure state
$\ket{0}^{A'} \ket{0}^{B'}$ from a known pure state $\rho^{AB}$
distributed between two parties $A$ and $B$. Here, $A'$, $B'$ are
subspaces of $A$, $B$ and the total amount of purity extracted is
$\log |A'| + \log |B'|$. The parties can borrow local pure ancilla,
apply local unitary operations and send a message from $A$ to $B$
through a dephasing channel. If local pure ancilla is borrowed, it
must be subtracted in order to properly account for the purity
extracted. We obtain the most efficient achievable bounds on one shot distributed
purity extraction, in terms of the rate of local ancilla borrowed by the protocol, while distilling pure qubits at the best known rate. Our protocols borrow little to no local
pure ancilla. Our bounds improve upon the existing bounds for this
problem in both one shot as well as asymptotic iid settings. In
particular they subsume all the asymptotic iid results of Devetak and
Krovi-Devetak. In addition, we derive upper bounds for the rate of distillation in the one shot setting, which nearly match our achievable bounds.
\end{abstract}

\section{Introduction}

Pure states are an important and ubiquitous resource in most quantum information processing protocols. Often, while implementing a quantum algorithm, one assumes the availability of pure states in the form of ancilla qubits which can be used as workspace for some computational operation. A specific example of this is the implementation of isometric operators as quantum gates in a circuit. Due to their widespread use, pure states are often assumed to be a freely available resource in most quantum information processing protocols. However, the question remains as to the cost one has to incur to prepare such pure states in the lab. Indeed, that this is a nontrivial operation was realised by Landauer \cite{Landauer}, who showed that to initialise an arbitrary classical bit to some preset value, an operation called erasure, one has to do work. Along a similar vein, the works of Bennett et al. and Szilard \cite{Bennet_etal_thermo, Szilard} prove that one can extract work from a thermal bath if the system is initialised to a pure state.

The above works underscore the importance of characterising the resources that are necessary to produce pure states in the lab. To that end, we consider the problem of \emph{purity distillation} and give an informal introduction below.

\subsection*{An Informal Description of Purity Distillation}

The problem of purity distillation is concerned with characterising the rate at which pure qubit states can be obtained from a given known input state, using certain admissible quantum operations. To that end, we first require a \emph{measure of purity} of a given known state, and a list of allowable operations under which our chosen measure of purity does not increase. For a given state $\rho^A$, a good choice for the measure of purity found in the literature \cite{Oppenheim_puritydil, GOUR, Streltsov_2018} is the following:
\[
\log \abs{A}-\wtt{H}(A),
\]
where $\wtt{H}(\cdot)$ is a placeholder for a well defined notion of entropy of a state which is suitable for our purposes. A list of allowable operations then takes the following form:
\begin{enumerate}
    \item Tracing out a subsystem.
    \item Appending a maximally mixed state in some register $A_{\mix}$, i.e. we allow access to private randomness.
    \item A special class of quantum operations $\brak{\mathcal{N}}$.
\end{enumerate}

The special class of quantum operations depends on the generality in which one wishes to treat the theory of purity distillation, and different authors have used increasing larger sets of operations in their treatment of the topic (see \cite{Streltsov_2018}). What is important is that the measure of purity should be non-increasing under these sets of operations.

Looking ahead, one can add a further allowable operation to the above list, that of borrowing pure ancilla qubits locally in some register $C_{\pure}$. However, since this clearly increases the measure of purity, one must account for this by modifying the formula of the measure of purity by the term $-\log \abs{C_{\pure}}$. We refer to this as operation as `borrowing pure ancilla qubits in a catalytic manner'.

Given the above setting, we can now define the tasks of \emph{local} and \emph{distributed} purity distillation: \\
\vspace{1mm}
\par \noindent {\bf Local Purity Distillation} : Given a quantum state $\rho^A$, a party Alice can use any finite sequence of operations from the list of allowed operations to produce a state $\sigma^{A_p}$, such that:
\[
\norm{\sigma^{A_p}-\ketbra{0}^{A_p}}_1\leq \eps,
\]
for some error parameter $\eps$. The goal is to maximise $\log \abs{A_p}$. \\
\vspace{1mm}
\par \noindent {\bf Distributed Purity Distillation} : Given a quantum state $\rho^{AB}$, where the party Alice has access to the system $A$ and the party Bob has access to the system $B$, and a completely dephasing channel $\mathcal{P}^{X_A\to X_B}$ from Alice to Bob, the parties are can use any finite sequence of local allowable operations, together with one classical message from Alice to Bob via the completely dephasing channel, to produce a joint state $\sigma^{A_pB_p}$, such that:
\[
\norm{\sigma^{A_pB_p}-\ketbra{0}^{A_p}\otimes \ketbra{0}^{B_p}}_1\leq \eps,
\]
for some error parameter $\eps$. The goal is to maximise $\log \abs{A_p}+\log \abs{B_p}$.

A further generalisation of the task of distributed purity distillation (DPD) is when we also charge for the amount of classical communication from Alice to Bob. To that end, we consider (informally) the following task:\\
\vspace{1mm}
\par \noindent {\bf DPD with Bounded Classical Communication} : Given the setting of DPD, the goal is to maximise the quantity $\log \abs{A_p}+\log \abs{B_p}$ with the additional constraint that the number of classical bits that Alice is allowed to send to Bob is at most $C_{\textup{classical}}$.\\
\vspace{1mm}
\par We will now make some remarks regarding the choice of special quantum operations in the set of allowable operations. From an informal perspective, it would seem logical that one should only allow local unitary operations in the set of special quantum operations $\brak{\mathcal{N}}$. This is because any other quantum operation would require additional ancilla qubits to be implemented in the lab using quantum circuits. Nevertheless, one may ask whether the set of special quantum operations may be enlarged from the set of unitary operators, to include maps which do not increase the measure of purity of the input state. Note that if one is allowed to include any such non-unitary map in the set $\brak{\mathcal{N}}$, its use will be \emph{free}, in the sense we will not charge for the number of ancilla qubits required to implement this map as a quantum circuit. Indeed, this topic has been studied in the works \cite{GOUR, Streltsov_2018}, where the authors show that their choice of purity measure does not increase under the action of unital CPTP maps.

However, the choice of $\brak{\mathcal{N}}$ in this paper is more restrictive and informed by a distinction between the tasks of local and distributed purity distillation. To see this distinction, note that one must allow some communication channel from Alice to Bob in the distributed setting. If not, then the best the parties can do is two locally optimal protocols on their systems $A$ and $B$. The quantum operator used as the communication channel must be a member of the set $\brak{\mathcal{N}}$. However, the choice of this channel cannot be an arbitrary operator from $\brak{\mathcal{N}}$. For example, if one allows the identity superoperator on the system $A$, which is a unital CPTP map, to be used as a channel, Alice can then send her entire system $A$ to Bob This trivially reduces the distributed distillation problem to the local distillation problem. Thus, we must \emph{fix} a choice of channel in the distributed case.

A judicious choice of channel is a \emph{classical} communication channel from Alice to Bob, modelled by the completely dephasing map $\mathcal{P}^{X_A\to X_B}$. Informally, we do not allow Alice to send any entangled bits to Bob, but allow classical communication. Note that this map is also a unital CPTP. This choice can be justified from a practical perspective as well,  given that robust quantum channels across large distances have not yet been realised.

Given the above discussion, throughout this paper we will fix our set of special operations $\brak{\mathcal{N}}$ to include only unitary operators and completely dephasing maps. In addition, as pointed out above, we allow partial trace, completely mixed ancilla and borrowing local pure ancilla catalytically. We will show that our choice of purity measure is non-increasing under the set of allowable operations which we choose to work with. In fact the same set of operations were allowed by the earlier works of \cite{Devetak_purity, KroviDevetak}. The reader is referred to Sections \ref{sec:defpurity} and \ref{sec:localprotocols} for the rigorous definitions and lemmas pertaining to the discussion above.

\begin{remark}
    We remark that our choice of purity measure is non-increasing even under the action of unital CPTP maps. This is easily seen from the proofs presented in Section \ref{sec:localprotocols}. However, we do not comment on this further to focus on the main contribution of this paper.
\end{remark}

\subsection*{History and Previous Works}

The problems of local and distributed purity distillation first appeared in the works \cite{Oppenheim_puritydil, Horodecki_puritydilution, Synak-purity}. Specifically, the distributed distillation problem was first introduced in the asymptotic iid setting \cite{Oppenheim_puritydil}, and some preliminary bounds for the case when both $1$-way and $2$-way communication is allowed between Alice and Bob was given in \cite{Horodecki_puritydilution} in the CLOCC (closed local operations and classical communication) setting in the asymptotic iid regime. In this setting the parties are not allowed to borrow any ancilla qubits catalytically, nor do they have access to private randomness. This implies that for an input state $\rho^A$, the set of operations $\brak{\mathcal{N}}$ are allowed to be \emph{only} unitary operators on $A$ along with completely dephasing maps.

 A further generalisation of this setting where the parties have access to private randomness, abbreviated as NLOCC (noisy local operations and classical communication) was also considered in \cite{Horodecki_puritydilution}. In this case one allows local unitaries to act on both the input register as well as the system which holds the completely mixed state. The set of quantum operations $\brak{\mathcal{N}}$ is clearly larger in this case than CLOCC, since one can construct operations on the input register which are convex combinations of unitary operators. However, this model still does not allow the parties to borrow ancilla qubits catalytically. A tight lower bound for the local distillation problem was provided in \cite{HorodeckiLB}, in the asymptotic iid setting.

Aside from the preliminary works mentioned above, the first detailed treatment for the purity distillation problem appeared in the work of Devetak \cite{Devetak_purity}. Devetak was the first to introduce the idea of borrowing pure ancilla in a catalytic manner, formalised as the CLOCC${}^\prime$ paradigm. In this paradigm one is allowed to borrow pure local ancilla qubits, but has to discount them from the final expression for number of pure qubits distilled. This relaxation allowed Devetak to characterise the rate of distributed purity distillation, when unbounded one-way classical communication is allowed. 
In particular, Devetak showed that, given $n$ iid copies of a bipartite state $\rho^{AB}$, where $A$ and $B$ are shared between two parties, and unbounded one-way classical communication, it is possible to distil pure qubits at (roughly) a rate:
\begin{align}\label{eq:ratedevetak}
\log \abs{A}-H(A)+\log \abs{B}-H(B)+\frac{1}{n}\max\limits_{\Lambda_n} I(X^n : B^n)
\end{align}
for a large enough $n$, where $\Lambda_n$ is a rank-$1$ POVM that acts on the system $A^n$ to produce a classical register $X_n$. It was also shown in the same paper that in the case of unbounded classical communication, this bound is tight in the iid limit.

The reader may have guessed that the additive mutual information term appearing in the expression above contributes a surplus of pure distillable qubits, more so than what a na\"ive application of two local protocols on the $A$ and $B$ systems would have allowed. As we shall see shortly, these surplus pure qubits are distilled by using the \emph{classical-quantum} correlations between the two systems $A$ and $B$ (see \cite{Devetak_Winter} for more details). These correlations are extracted during the protocol execution by using the POVM $\Lambda_n$.  Since the mutual information quantity above is maximised by rank-$1$ POVM, Devetak only considers these and indeed his protocol and proof techniques are heavily reliant on this fact.

The problem of DPD with bounded classical communication in the asymptotic iid setting was first considered by Krovi and Devetak in \cite{KroviDevetak}, where the authors not only provided tight upper and lower bounds for this problem but also significantly simplified the original proof given in \cite{Devetak_purity}. In fact, the authors of that paper showed that under the constraint that Alice is allowed to send at most $nC_{\textup{classical}}$ number of bits to Bob, it is possible to recover pure states from $\rho^{AB}$ at a rate similar in expression to the formula in Equation \ref{eq:ratedevetak}, with the important distinction that maximisation is now over the set of all POVMs $\Lambda_n$ such that $I(X_n : B^nR^n)\leq nC_{\textup{classical}}$. Here the mutual information is computed with respect to the post measurement state $\rho^{X_nB^nR^n}\coloneqq \left(I^{R^nB^n}\otimes \Lambda_n\right)\left((\ketbra{\rho}^{ARB})^{\otimes n}\right)$ and $\ket{\rho}^{ABR}$ is some purification of the shared state $\rho^{AB}$. Note that in this case the POVM $\Lambda_n$ will in general no longer be rank-$1$.

All the works mentioned above tackle the problem of purity distillation in the asymptotic iid setting, that is, when one assumes that many independent copies of the resources are available to the parties in the protocol. Recently, Chakraborty, Nema and Buscemi \cite{CNB23} presented one-shot versions of the local and distributed purity distillation protocols in \cite{CNB23}, where the authors assumed that only \emph{one} copy of the underlying state is available to the parties taking part in the protocol. Although the techniques presented in that paper generalise Devetak's \cite{Devetak_purity} original techniques to the one-shot setting, it is not immediately clear how one can adapt them to the case of DPD with bounded communication. In particular it is not clear how one can extend the asymptotic iid results of Krovi Devetak on DPD with bounded communication to the one shot setting.

\subsection*{Our Contribution}
All the protocols for distributed purity distillation mentioned so far work in the paradigm where one is allowed to borrow some ancilla qubits at the beginning of the protocol but must account for them in the final rate. In fact, all the existing protocols which achieve the best known rate for this problem, whether in the asymptotic iid setting (\cite{Devetak_purity} and \cite{KroviDevetak}) or the one shot setting (\cite{CNB23}) crucially require ancilla qubits which they use in this catalytic manner. Furthermore, the rate at which these protocols borrow ancilla is typically quite high, roughly $\frac{1}{n}I(X_n:R^nB^n)$ for the asymptotic iid protocols and $I_{\max}^{\eps}(X:RB)$ for one shot protocols, where the mutual information quantities are always computed with respect to the post measurement state $\rho^{XRB}$ obtained after the action of the POVM. This is clearly undesirable from a practical standpoint, since one would hope that protocols used to distil pure qubits would themselves require only a few initial pure qubits to function.

Note that there are ad-hoc techniques, called bootstrapping, to reduce the rate of pure qubits which the protocol consumes in the asymptotic iid setting. Indeed, given $n$ iid copies of the underlying state, one can divide these states into blocks of size $\sqrt{n}$. One can then use some ancilla to run the Krovi-Devetak protocol on the first block, and recover this ancilla at the end of the protocol. The recovered ancilla can then be used catalytically on subsequent runs of the Krovi-Devetak protocol on the other $\sqrt{n}$ sized blocks. There may be other similar strategies which use the idea of dividing the iid states into smaller blocks to reduce the number of ancilla qubits that the parties have to borrow (see Devetak \cite{Devetak_purity}) However, these strategies are ad-hoc and depend upon assumptions regarding the number of pure qubit states that each party can distil. Further, these techniques completely fail in the one-shot setting where only \emph{one} copy of the underlying state is available, and one cannot do any bootstrapping.

In  this paper we present a uniform approach towards distilling the maximum number of pure qubits in the distributed setting for a given amount of classical communication, while at the same time
reducing the number of initial pure catalytic ancilla qubits borrowed, which works both in the one-shot and the asymptotic iid setting. We call the proposed  protocol as \name (see Theorem \ref{thm:informal} in Section \ref{sec:optWithoutAncilla}). In comparison with existing protocols, \name has several key improvements with regard to the number of ancilla qubits it requires, while maintaining the same rate of distillation as existing protocols. To highlight these improvements, we present a comparison of \name with the currently existing protocols, both in the asymptotic iid and one shot settings. We show that \name offers advantages over existing protocols in both paradigms:

\begin{enumerate}
    \item In the asymptotic iid limit, the rate at which \name requires input ancilla qubits to function is {\bf $0$}, independent of the input mixed state $\rho^{AB}$ or the input POVM $\Lambda_n$.
    \item In the one shot setting, when unbounded $1$-way classical communication is allowed  (i.e. the setting of Devetak's original paper), \name requires at most $O(\log \frac{1}{\eps})$ pure ancilla qubits in the worst case. In comparison, the one shot protocol of \cite{CNB23} requires roughly $I_{\max}^{\eps}(X:RB)$ ancilla qubits to work under similar assumptions. See Corollary \ref{corol:unboundedcommcompare} in Section \ref{sec:compare} for details.

    \item To facilitate a fair comparison in the more general scenario when the rate of classical $1$-way communication is bounded in the one shot setting, we first present an appropriate one shot generalisation of the original Krovi-Devetak protocol \cite{KroviDevetak}, which we call \kd (Section \ref{sec:kd}). We prove that under mild conditions, \name requires provably fewer ancilla qubits to function than \kd. See Corollary \ref{corol: general compare} in Section \ref{sec:compare} for details.
\end{enumerate}

The main technical ingredient in the construction of \name is an embedding technique which we use to simulate the action of the POVM $\Lambda$ on the $A$ space without requiring too many extra ancilla qubits. Note that since the parties are allowed \emph{only} local unitary operations, any POVM must be implemented coherently. To implement the POVM $\Lambda$ one would then require an extra register to store the classical outcomes, i.e., given any POVM $\Lambda$ its coherent counterpart can be expressed as the isometry $\sum\limits_{x}\ket{x}^X\sqrt{\Lambda_x}^A$, where the log dimension of the system $X$ is precisely the number of intial pure ancilla qubits required.

Note that since $\Lambda$ is an arbitrary POVM, one cannot hope to bound the number of possible outcomes that this POVM has. The first step therefore is to replace $\Lambda$ with another POVM $\widetilde{\Lambda}$ which has far fewer outcomes (typically $2^{I_{\max}^{\eps}(X:RB)}$ many) but which nevertheless preserves the correlations between the the classical output register and the system $B$. This step is made possible by the measurement compression theorem of Winter \cite{Winter_meascomp} and has been used by  both Devetak \cite{Devetak_purity} and Krovi and Devetak \cite{KroviDevetak}. The problem is harder in the one shot setting but one can use a recent one shot measurement compression theorem of \cite{ChakrabortyPadakandlaSen_22} to get around it. This theorem is the key to results presented in \cite{CNB23}.

Note however, that even after bounding the number of outcomes of the POVM, one still needs to borrow some $I_{\max}^{\eps}(X:RB)$ initial pure ancilla qubits to store the outcomes. Our main contribution goes towards reducing this rate as much as possible. The main idea is that we design a unitary operator $U_{\wtt{\Lambda}}$ which simulates the action of measuring the $A$ register coherently with $\widetilde{\Lambda}$ \emph{in place}, i.e., in the $A$ register itself, requiring very little, and in many cases zero, additional pure ancilla in order to store the measurement outcomes. The details of this are technical and can be found in Section \ref{sec:optWithoutAncilla}.

Aside from our main result, we also present upper bounds on the rates of purity distillation that \emph{any} local and distributed 
distillation algorithm can hope to achieve in the one shot setting. Prior to our work such bounds were not known in the one shot regime. These upper bounds nearly match the rate of distillation given by \kd and \name.

\subsubsection*{Organisation of the paper}

The paper is organised as follows: in Section \ref{sec:def} we present the definitions of the one-shot quantities that we use throughout the paper. We also state several known properties of these quantities as facts and prove some other relevant properties as lemmas in this section. In Section \ref{sec:defpurity} we formally present the definitions of $\eps$-purity and the tasks of local and distributed purity distillation. In Sections \ref{sec:localprotocols} and \ref{sec:UBdistributed} we present upper bounds pertaining to the tasks of local and distributed purity distillation. We should mention that Section  \ref{sec:localprotocols} also contains details regarding an optimal achievable protocol for local purity distillation. We would also like to highlight Section \ref{sec:unbounded} in which we focus on upper bounds for distributed purity distillation in the case when unbounded classical communication is allowed. In Section \ref{sec:LBdistributed} we derive lower bounds for distributed purity distillation in the case of bounded communication, with the \kd protocol presented in Section \ref{sec:kd}. We present our main result, the existence of \name, in Section \ref{sec:optWithoutAncilla}. Section \ref{sec:compare} contains a comparison between \name and \kd.

\section{Preliminaries: Relevant Quantities}\label{sec:def}

\subsection{Notation and Some Basics}

\begin{definition}[{\bf Quantum State}] A quantum state $\rho$ on some register $A$ is a positive semi-definite matrix with trace $1$.
\end{definition}

\begin{remark}
    The notation $\sigma \geq 0$, for a matrix $\sigma$ is used to denote the fact that $\sigma$ is positive semi-definite. More generally, $\rho \geq \sigma$ implies that the matrix $\rho-\sigma\geq 0$. This partial order of the positive semi-definite matrices is referred to as the Loewner order.
 \end{remark}

 \begin{definition}[{\bf Fidelity and Generalised Fidelity}]
     Given two quantum states $\rho$ and $\sigma$, the fidelity between the two states is defined as:
     \[
     F(\rho, \sigma)\coloneqq \norm{\sqrt{\rho}\sqrt{\sigma}}_1.
     \]
     The fidelity can be extended in a meaningful way to matrices which are sub-states, i.e. matrices $\rho$ and $\sigma$ such that $0 \leq \rho, \sigma, \leq I $, in the following way:
     \[
     \overline{F}(\rho, \sigma)\coloneqq F(\rho, \sigma)+\sqrt{(1-\Tr[\rho])(1-\Tr[\sigma])}.
     \]
     $\overline{F}(\cdot, \cdot )$ is referred to as the generalised fidelity.
 \end{definition}

 \begin{remark}
 The generalised fidelity was defined in \cite{dualitysmoothminmax}.
 \end{remark}

 \begin{definition}{\bf $\cdot$ Operation}
    Given an operator $M^{A\to B}$ and the operator $N^A$, we define the $\cdot$ operation as follows:
    \[
    M\cdot N\coloneqq MNM^{\dagger}.
    \]
\end{definition}

\subsection{Definitions: One-Shot Entropic Quantities}
In this section we introduce the one-shot entropic quantities which we will be  using in the subsequent sections to describe our protocols.

\begin{definition}[\textbf{Smoothed Support Max Entropy}]\label{def:hmaxtilde}
Given a quantum state $\rho^A$, let us denote its eigenvalues by $\lambda_1,\ldots \lambda_{\abs{\textup{supp}(\rho)}}$ (in ascending order) corresponding to the eigenvectors $ v_1, \ldots, v_{\abs{\textup{supp}(\rho)}}$. Let $\lambda_1, \ldots \lambda_k$ denote the smallest eigenvalues such that $\sum_{i} \lambda_i \leq \eps$. We define the {\bf smoothed support max entropy} of $\rho^A$ as 
\[
\thmax^{\eps}(A)_{\rho}\coloneqq \log \left(\abs{\textup{supp}(\rho)}-k\right).
\]
\end{definition}

\begin{definition}[\textbf{Smoothed Norm Max Entropy}]\label{def:hmax}
Given the setup of in Definition \ref{def:hmaxtilde}, we define the {\bf smoothed norm max entropy} of the state $\rho^A$ as
\[
\phmax^{\eps}(A)_{\rho}\coloneqq \log \frac{1}{\lambda_{k+1}}
\]
\end{definition}

\begin{definition}[\textbf{Conditional Smooth Hypothesis Testing Entropy}]
Given a quantum state $\rho^{AB}$ we define the {\bf Smooth Hypothesis Testing Entropy} as
\begin{align*}
    H_H^{\eps}(A|B)_{\rho} &\coloneqq -D_H^{\eps}(\rho^{AB}~||~\mathbb{I}^A\otimes \rho^{B}) 
\end{align*}
where $D_H^{\eps}$, the hypothesis testing relative entropy, is defined as:
\[
2^{-D_H^{\eps}(\rho||\sigma)}\coloneqq \min\limits_{\substack{0\leq \Pi\leq \mathbb{I}\\ \Tr\left[\Pi\rho\right]\geq 1-\eps}} \Tr\left[\Pi\sigma\right].
\]
\end{definition}

For the couple of definitions that follow we will require the notion of an $\eps$-ball around a state $\rho$. The following definition can be found in \cite[Definition 10]{dualitysmoothminmax}:

\begin{definition}
    Given a quantum state $\rho^A$, we define the $\eps$-ball $\mathcal{B}^{\eps}(\rho)$ as:
    \[
    \mathcal{B}^{\eps}(\rho)\coloneqq \brak{\tau~:~0\leq \tau, \Tr\left[\tau\right]\leq 1, P(\tau, \rho)\leq \eps},
    \]
    where $P(\cdot, \cdot)$ is the purified distance on the space of sub-normalised states (see \cite{dualitysmoothminmax} for details).
\end{definition}

\begin{definition}{{\bf (Conditional Smooth Max Entropy)}}
Given a bipartite quantum state $\rho^{AB}$, we define the {\bf smooth max entropy } as
\[
H_{\max}^{\eps}(A|B)_{\rho} \coloneqq \min\limits_{\rho'\in \mathcal{B}^{\eps}(\rho)}\max\limits_{\substack{\sigma^B\geq 0\\ \Tr[\sigma]=1}} 2\log F\left(\rho'^{AB}, \textup{\one}^A\otimes \sigma^B \right)
\]

\end{definition}
\begin{definition}{{\bf (Conditional Smooth Min Entropy)}}
    Given a bipartite quantum state $\rho^{AB}$, the {\bf conditional smooth min entropy} is defined as:
    \[
    H_{\min}^{\eps}(A|B)_{\rho}\coloneqq \max\limits_{\rho'\in \mathcal{B}^{\eps}(\rho)}-\log \min\brak{\Tr\left[\sigma^B\right]~|~\sigma^B\geq 0, \rho'^{AB}\leq \mathbb{I}^A\otimes \sigma^B}.
    \]
\end{definition}
\subsection{Properties: One-Shot Entropic Quantities}

\begin{fact}{{\bf (Data Processing Inequality for the Smooth Min and Max entropies, \cite{dualitysmoothminmax})}}
    Given a state $\rho^{AB}$ and $\eps>0$, a CPTP map $\mathcal{E}^{B\to D}$, we define $\sigma^{AD}\coloneqq (\I^A\otimes \mathcal{E}^B)(\rho^{AB})$. Then it holds that:
    \begin{align*}
        &H_{\min}^{\eps}(A|B)_{\rho} \leq H_{\min}^{\eps}(A|D)_{\sigma} \\
        &H_{\max}^{\eps}(A|B)_{\rho} \leq H_{\max}^{\eps}(A|D)_{\sigma}.
    \end{align*}
\end{fact}
We will also require the following data processing type inequality, presented in \cite{Tomamichel_thesis}. We present a simplified version of the original result, which is much more pertinent for our purposes:
\begin{fact}{{\bf (\cite{Tomamichel_thesis}}}\label{fact:unitalhmax}
     Given a state $\rho^{AB}$ and $\eps>0$, a unital CPTP map $\mathcal{E}^{A\to C}$, we define $\sigma^{CB}\coloneqq (\mathcal{E}^A\otimes \I^B)(\rho^{AB})$. Then it holds that:
    \begin{align*}
        &H_{\min}^{\eps}(A|B)_{\rho} \leq H_{\min}^{\eps}(C|B)_{\sigma} \\
        &H_{\max}^{\eps}(A|B)_{\rho} \leq H_{\max}^{\eps}(C|B)_{\sigma}.
    \end{align*}
\end{fact}
We refer below to a subset of the chain rules for the smooth min and max entropies, presented in \cite{Dupuis_chainrules}, that will prove useful for our purposes. We present the chain rules in a simplified form which is most pertinent for us.
\begin{fact}{{\bf (Chain Rules for Smooth Min and Max Entropy, \cite{Dupuis_chainrules})}}\label{fact:dupuischainrules}
    Given a state $\rho^{AB}$, it holds that:
    \begin{align*}
        &H_{\max}^{\eps}(AB)_{\rho}\geq H_{\min}^{O(\eps)}(A|B)_{\rho}+H_{\max}^{O(\eps)}(B)_{\rho}-O(\log \frac{1}{\eps}) \\
        &H_{\max}^{\eps}(AB)_{\rho}\geq H_{\max}^{O(\eps)}(A|B)_{\rho}+H_{\min}^{O(\eps)}(B)_{\rho}-O(\log \frac{1}{\eps}).
    \end{align*}
\end{fact}

\begin{fact}[\cite{CNB23}]\label{fact:tildemaxHrelation}
For any quantum state $\rho^A$ it holds that
\[
H_{\max}^{\eps}(A)_{\rho}\leq \thmax^{\eps}(A)_{\rho} \leq  \phmax^{\eps}(A)_{\rho} \leq \log\frac{\abs{A}}{\eps}
\]
\end{fact}

\begin{lemma}\label{lem:purehh}
Given a state $\rho^A$ and an arbitrary purification $\ket{\rho}^{RA}$, it holds that
\[
H_H^{\eps}(A)_{\rho}=H_H^{\eps}(R)_{\rho}.
\]
\end{lemma}
\begin{proof}
We will first show that we can assume that the optimising operator in the definition of $H^{\eps}_H(A)_{\rho}$ commutes with $\rho$. Let this operator be $\Pi$. To begin, consider the Schmidt decomposition of $\ket{\rho}^{AR}$:
\[
\ket{\rho^{AR}}= \sum_a\sqrt{P_A(a)}\ket{a}^A\ket{\zeta_a}^R,
\] which implies that
\[
\rho^{A}=\sum_a P_A(a) \ketbra{a}^A.
\]
Then, 
\[\begin{aligned}
\Tr\left[\Pi\rho\right] &= \sum_a P_A(a)\braket{a|\Pi|a} \\
\Tr\left[\Pi\right] &= \sum_a \braket{a|\Pi|a}
\end{aligned}
\]
Without loss of generality we can assume that $\Pi$ has non-negative eigenvalues only on a subspace of the support of $\rho$. Now, consider the operator:
\[
\widetilde{\Pi}\coloneqq \sum_{\ketbra{a} \in \textup{supp}(\rho)} \braket{a|\Pi|a} \ketbra{a}
\]
It is easy to see that $\widetilde{\Pi}$ has all the properties of $\Pi$ that we require. Thus we can assume that the optimising operator commutes with $\rho$. Next, 
we wish to compute the quantity $H_H^{\eps}(R)_{\rho}$. As before, we can assume that the optimising operator commutes with $\rho^R$, i.e., it diagonalises in the basis $\brak{\ket{\zeta_a}^A}$ and has non-negative eigenvalues only on a subspace of the support of $\rho^R$. This implies that, the optimising operator , say $\Sigma$, can be written as:
\[
\Sigma^R= \sum_{\zeta_a\in \textup{supp}(\rho^R)} \lambda_a \ketbra{\zeta_a}^R
\]
Finally, we see that the definition of $H_H^{\eps}(R)_{\rho}$ reduces to solving the following LP:
\begin{align*}
    \min &\sum_a \lambda_a \\
    &\sum_a P_A(a)\lambda_a \geq 1-\eps \\
    & 0\leq \lambda_a\leq 1
\end{align*}
It is not hard to see that this same LP that defines $H_H^{\eps}(A)_{\rho}$, if only we replace $\lambda_a$ with $\mu_a\coloneqq \braket{a|\Pi|a}$. Thus, it holds that
\[
H_H^{\eps}(R)_{\rho}=H_H^{\eps}(A)_{\rho}.
\]
This concludes the proof.
\end{proof}

\begin{lemma}\label{lem:puretensorsame}
    For a state $\rho^A$ and any pure state $\ket{\phi}^B$, it  holds that:
    \[
    H_H^{\eps}(AB)_{\rho\otimes \phi}= H_H^{\eps}(A)_{\rho}.
    \]
\end{lemma}

\begin{proof}
Note that the following holds since $\phi^B$ is pure:
\[
\rho\otimes \ketbra{\phi}=\sum_aP_A(a)\ketbra{a}^A\otimes \ketbra{\phi}^B.
\]
Therefore, from the proof of Lemma \ref{lem:purehh} we can see that $H_H^{\eps}(AB)_{\rho\otimes \phi}$ is given by the log of the solution of the LP:
\begin{align*}
\min &\sum\limits_{a} \lambda(a) \\
&\sum\limits_{a} P_A(a)\lambda(a)\geq 1-\eps.
\end{align*}
However, this is the same LP that gives the expression for $H_H^{\eps}(A)_{\rho}$. This concludes the proof.
\end{proof}

\begin{lemma}[{\bf Equivalence of the Smoothed Norm Max and Smooth Hypothesis Testing Entropies}]\label{lem:equivHHtildeHmax}
For any quantum state $\rho^A$ it holds that
\[
\thmax^{\eps}(A)_{\rho}-1 \leq H_H^{\eps}(A)_{\rho}\leq \thmax^{\eps}(A)_{\rho}
\]
\end{lemma}

\begin{proof}
To prove this lemma, we first observe that without loss of generality we can assume that the optimising operator for $H_H^{\eps}(A)_{\rho}$ diagonalises in the same basis as $\rho^A$. To see this, we argue via contradiction. Suppose the assumption isn't true. Let $\rho^A=\sum\limits_{a}p(a) \ketbra{v_a}^A$. Then, by definition:
    \begin{align*}
        &\Tr\left[\Pi^A\rho\right]\\
        = & \sum\limits_{a}p(a) \braket{v_a|\Pi|v_a}\\\
        \geq & 1-\eps.
    \end{align*}
    Since $0\leq \Pi\leq \I$, it holds that for all $a$, $0\leq \braket{v_a|\Pi|v_a}\leq 1$. We can then define a new operator $\Pi_{\textsc{opt}}$ whose eigenbasis contains the vectors $\brak{\ket{v_a}}$ (it can have more eigenvectors since the rank of $\Pi$ may be larger than the rank of $\rho$), and which has the same eigenvalues as $\Pi$. Clearly, $\Pi_{\textsc{opt}}$ satisfies the criteria that $\Tr\left[\Pi_{\textsc{opt}}\rho\right]\geq 1-\eps$, and also $\Tr\left[\Pi_{\textsc{opt}}\right]=\Tr\left[\Pi\right]$. Therefore, we can always assume that the optimiser for $H_H^{\eps}(A)_{\rho}$ commutes with $\rho$. This immediately implies the upper bound since we obtain $\thmax^{\eps}(A)_{\rho}$ by projecting onto all but those eigenvectors of $\rho^A$ whose eigenvalues are the smallest and add up to at most $\eps$. 

Now, since we know that the optimising operator for $H_H^{\eps}(A)_{\rho}$ diagonalises in the same basis as $\rho^A$, once can assume that the following holds for all such candidate operators $\Pi^A$:
\[
\begin{aligned}
\rho^A &= \sum\limits_{a} P_A(a) \ketbra{a}^A \\
\Pi^A &= \sum\limits_{a} \lambda(a) \ketbra{a}^A
\end{aligned}
\]
Then, it holds that the problem of finding $H_H^{\eps}(A)_{\rho}$ can be reduced to solving the following LP:
\[
\begin{aligned}
\min & \sum\limits_{a} \lambda(a) \\
& \sum\limits_{a} P_A(a)\lambda(a) \geq 1-\eps
\end{aligned}
\]
We know from \cite{Pranab_notes} that the log of the solution to this LP is at least $\thmax^{\eps}(A)_{\rho}-1$. The lower bound follows. This concludes the proof.
\end{proof}

\begin{lemma}{\bf (Subadditivity of the Smooth Hypothesis Testing Entropy)}
Given a bipartite quantum state $\rho^{AB}$, it holds that
\[
H_H^{3\sqrt{\eps}}(AB)_{\rho} \leq H_H^{\eps}(A)_{\rho}+H_H^{\eps}(B)_{\rho}.
\]
\end{lemma}

\begin{proof}
To see that this holds, let $\Pi^A$ and $\Pi^B$ be the optimising operators  for $H_H^{\eps}(A)_{\rho}$ and $H_H^{\eps}(B)_{\rho}$ respectively. Then,
\begin{align*}
    \Tr\left[\Pi^A\otimes \Pi^B \rho^{AB}\right] &= \Tr \left[ \left(\sqrt{\Pi^A}\otimes \sqrt{\Pi^B} \right)\cdot \rho^{AB} \right] \\
    &= \Tr \left[\left(\I^A\otimes \sqrt{\Pi}^B\right)\cdot\left(\sqrt{\Pi}^A\otimes \I^B\cdot \rho^{AB}-\rho^{AB}\right)\right]+\Tr\left[\left(\I^A\otimes \sqrt{\Pi}^B\right)\cdot \rho^{AB}\right] \\
    &\geq 1-\eps - \norm{\sqrt{\Pi}^A\cdot\rho^{AB}-\rho^{AB}}_1
\end{align*}
We know that
\[
\Tr\left[\Pi^A\otimes \I^B \rho^{AB}\right] \geq 1-\eps
\]
By the Gentle Measurement Lemma, we can then see that
\[
\norm{\sqrt{\Pi}^A\cdot\rho^{AB}-\rho^{AB}}_1 \leq 2\sqrt{\eps}
\]
Therefore we can conclude that,
\[
\Tr\left[\Pi^A\otimes \Pi^B \rho^{AB}\right] \geq 1-3\sqrt{\eps}
\]
Therefore, $\Pi^A\otimes \Pi^B$ is a candidate optimiser for $H_H^{3\sqrt{\eps}}(AB)$, which implies that
\[
H_H^{3\sqrt{\eps}}(AB)_{\rho} \leq H_H^{\eps}(A)_{\rho}+H_H^{\eps}(B)_{\rho}.
\]
This concludes the proof.
\end{proof}

\begin{lemma}\label{lem:additivehh}
    Given the quantum state $\rho^A\otimes \pi^B$, where $\pi^B$ is the maximally mixed state on the system $B$, it holds that:
    \[
    H_H^{\eps}(AB)_{\rho\otimes \pi}=H_H^{\eps}(A)_{\rho}+\log \abs{B}.
    \]
\end{lemma}

\begin{proof}
    By the arguments used in the proof of Lemma \ref{lem:equivHHtildeHmax}, we know that $H_H^{\eps}(AB)_{\rho\otimes \pi}$ is obtained by solving the following LP, which we call $LP1$:
    \begin{align*}
        \min & \sum\limits_{a,b} \lambda(a,b) \\
& \sum\limits_{a,b} \frac{P_A(a)}{\abs{B}}\lambda(a,b) \geq 1-\eps,
    \end{align*}
    where $\rho^A=\sum\limits_{a}P_A(a)\ketbra{a}^A$ and $\pi^B=\sum\limits_{b}\frac{1}{\abs{B}}\ketbra{b}^B$. 
   Consider also the following LP, which we call $LP2$:
    \begin{align*}
        \min & \abs{B}\cdot \left(\sum\limits_{a} \lambda(a)\right) \\
& \sum\limits_{a} P_A(a)\lambda(a) \geq 1-\eps.
    \end{align*}
    Now note that an optimising $\brak{\lambda(a)}$ for $LP2$ can be turned into a \emph{feasible} $\brak{\lambda(a,b)}$ for $LP1$ by simply declaring $\lambda(a,b)=\lambda(a), \forall b$. Similarly, an optimising $\brak{\lambda(a,b)}$ for $LP1$ can be turned into a \emph{feasible} $\brak{\lambda(a)}$ for $LP2$ by defining:
    \[
    \lambda(a)\coloneqq \sum\limits_{b}\frac{1}{\abs{B}}\lambda(a,b).
    \] 
    This argument immediately implies that the minima of $LP1$ and $LP2$ are equal. However, note that the minima of $LP2$ is precisely:
    \[
    \abs{B}\cdot 2^{H_H^{\eps}(A)_{\rho}}.
    \]
    Collating all the arguments above implies that:
    \[
    H_H^{\eps}(AB)_{\rho\otimes \pi}=H_H^{\eps}(A)_{\rho}+\log \abs{B}.
    \]
    This concludes the proof.
\end{proof}

\begin{corollary}\label{corol:upperboundhhdim}
    Given a state $\rho^{AB}$, it holds that :
    \[
    H_H^{\eps}(AB)_{\rho}\leq H_H^{\eps}(A)_{\rho}+\log \abs{B}.
    \]
\end{corollary}

\begin{proof}
    From the theory of unitary $1$-designs we know that there exists a set of unitaries $U_i^B$ on the system $B$ and probability distribution $\brak{p_i}$ such that, for any matrix $M^B$, it holds that:
    \[
    \sum\limits_{i}p_i U_iMU_i^{\dagger}= \Tr\left[M\right]\frac{I^B}{\abs{B}}.
    \]
    Let us denote the operation $\sum\limits_{i}p_iU_i(\cdot)U_i^{\dagger}$ as $\mathcal{T}^{B\to B}$. Note that $\mathcal{T}$ is a unital CPTP map. Also, note that it is not hard to show that for any state $\rho^{AB}$, it holds that:
    \[
    \left(\I^A\otimes \mathcal{T}^B \right) (\rho^{AB}) = \rho^A\otimes \pi^B,
    \]
    where $\pi^B$ is the maximally mixed state on the system $B$.    Therefore, by Lemma\ref{lem:additivehh}, we see that:
    \begin{align*}
        H_H^{\eps}(AB)_{\rho^{AB}} &\leq H_H^{\eps}(AB)_{\left(\I^A\otimes \mathcal{T}^B \right)(\rho^{AB})} \\
        &= H_H^{\eps}(AB)_{\rho^A\otimes \pi^B}\\
        &= H_H^{\eps}(A)_{\rho}+\log \abs{B}.
    \end{align*}
    This concludes the proof.
\end{proof}

\begin{lemma}\label{lem:hhneg}
Let $\sigma^A$ be a state such that
\[
\norm{\sigma^A-\ketbra{0}^A}_1 \leq \eps
\]
Then
\[
H_H^{\eps}(A)_{\sigma} \leq 0
\]
\end{lemma}

\begin{proof}
The condition in the statement of the lemma implies that
\[
\braket{0|\sigma|0} \geq 1-\eps
\]
This implies that $\ketbra{0}^A$ is a valid candidate for the optimising operator for $H_H^{\eps}(A)_{\sigma}$. Since $\ketbra{0}$ has trace $1$, the result follows. This concludes the proof.
\end{proof}

\begin{lemma}\label{lem:averageToWorstcaseHH}
    Given a quantum cq state $\rho^{XB}=\sum\limits_{x} P_X(x)\ketbra{x}^X\otimes \rho_x^B$ where $x\in \mathcal{X}$, it holds that there exists a subset $\mathcal{S}\subseteq \mathcal{X}$ such that
    \begin{align*}
        & \Pr\limits_{P_X}\left[\mathcal{S}\right]\geq 1-2\sqrt{\eps} \\
        & H_H^{\sqrt{\eps}}(\rho^B_x)\leq H_H^{\eps}(B|X)_{\rho}-\log\eps, ~~~~\forall x\in \mathcal{S}.
    \end{align*}
\end{lemma}

\begin{proof}
    Without loss of generality we can assume that the optimising operator $\Pi^{XB}$ in the definition of $H_H^{\eps}(B|X)_{\rho}$ is of the form:
    \[
    \Pi^{XB}=\sum\limits_{x}\ketbra{x}\otimes \Pi_x^B.
    \]
    By definition, this operator has the property that:
    \[
    \sum\limits_{x}P_X(x)\Tr\left[\Pi_x^B\rho_x^B\right] \geq 1-\eps.
    \]
    By Markov's inequality, we can then see that there exists a set $\mathcal{T}_1$ such that $\Pr\limits_{P_X}\left[\mathcal{S}\right]\geq 1-\sqrt{\eps}$ and for all $x\in \mathcal{T}_1$,
    \[
    \Tr\left[\Pi_x^B\rho_x^B\right]\geq 1-\sqrt{\eps}.
    \]
    Again, by definition, it holds that:
    \[
    2^{H_H^{\eps}(B|X)_{\rho}}=\sum\limits_{x}P_X(x)\Tr\left[\Pi_x^B\right].
    \]
    Again, Markov's inequality tells us that there exists a set $\mathcal{T}_2\subseteq \mathcal{X}$ of probability (under $P_X$) of at least $1-\eps$ such that for all $x\in \mathcal{T}_2$, it holds that:
    \[
    \Tr\left[\Pi_x^B\right]\leq \frac{2^{H_H^{\eps}(B|X)_{\rho}}}{\eps}.
    \]
    Therefore, for all $x\in \mathcal{T}_1\bigcap \mathcal{T}_2$ (which has probability at least $1-2\sqrt{\eps}$ under $P_X$), it holds that $\Pi_x^B$ is a candidate for the optimiser in the definition of $H_H^{\sqrt{\eps}}(\rho_x^B)$. Thus defining $\mathcal{S}\coloneqq \mathcal{T}_1\bigcap \mathcal{T}_2$ we see that the result follows. This concludes the proof.
\end{proof}

\begin{lemma}\label{lem:hhzeroforconditionalpure}
Given a cq state 
\[\rho^{XB}= \sum\limits_{x}P_X(x)\ketbra{x}^X\otimes \ketbra{v_x}^B,\]
it holds that $H_H^{\eps}(B~|~X)_{\rho}\leq 0$.
\end{lemma}

\begin{lemma}\label{lem:hhswitchforconditionalpure}
    Given a cq state of the form 
    \[
    \rho^{XAB}=\sum\limits_{x}P_X(x)\ketbra{x}^X\otimes \ketbra{v_x}^{AB},
    \]
    it holds that $H_H^{\eps}(B~|~X)_{\rho}=H_H^{\eps}(A~|~X)_{\rho}$.
\end{lemma}

The proofs of Lemma \ref{lem:hhzeroforconditionalpure} and \ref{lem:hhswitchforconditionalpure} can be found in Appendix \ref{appendix:hhlemmas}. Finally, we will require a data processing inequality for $H_H^{\eps}$:
\begin{lemma}\label{lem:dataprocessinghh}
    Given a state $\rho^{AB}$, a CPTP map $\mathcal{E}^{B\to D}$, and a unital CPTP map $\mathcal{F}^{A\to C}$, it holds that:
    \[
    \begin{aligned}
    &H_H^{\eps}(A|B)_{\rho}\leq H_H^{\eps}(A|D)_{\left(\I^A\otimes \mathcal{E}^B\right)(\rho^{AB})}, \\
    &H_H^{\eps}(A|B)_{\rho}\leq H_H^{\eps}(C|B)_{\left(\mathcal{F}^A\otimes \I^B\right)(\rho^{AB})} .
    \end{aligned}
    \]
    
\end{lemma}
\begin{proof}
       The proof of the first inequality follows directly from the data-processing inequality for $D_H^{\eps}(\cdot||\cdot)$ and the definition of $H_H^{\eps}(A|B)_{\rho}$. For the second inequality, note that by definition, we know that:
    \[
    \exp\left(H_H^{\eps}(C|B)_{\left(\mathcal{F}^A\otimes \I^B\right)(\rho)}\right)= \min\limits_{\substack{\Pi^{CB}~:~0\leq \Pi^{CB}\leq I^{CB} \\ \hspace{3em}\Tr\left[\Pi\left(\mathcal{F}^A\otimes \I^B\right)(\rho^{AB})\right]\geq 1-\eps}} \Tr\left[\Pi^{CB}\left(I^{C}\otimes \rho^B\right)\right].
    \]
    Let $\Pi^*$ be the optimising operator in the expression of $H_H^{\eps}(C|B)_{\left(\mathcal{F}^A\otimes \I^B\right)(\rho)}$. We will show that the operator $(\mathcal{F}^{\dagger})^{C\to A}\otimes \I^{B}(\Pi^*)$ is a candidate optimiser for $H_H^{\eps}(A|B)_{\rho}$, where $\mathcal{F}^{\dagger}$ is the adjoint of $\mathcal{F}$. Firstly, note that since $\mathcal{F}$ is unital and completely positive, $\mathcal{F}^{\dagger}$ is trace preserving and completely positive, i.e., CPTP. Also, since $\mathcal{F}$ is trace preserving, $\mathcal{F}^{\dagger}$ is unital. Note that since CPTP maps preserve operator inequalities, it holds that:
    \begin{align*}
        0 &\leq \left(\mathcal{F}^{\dagger C}\otimes \I^B\right)(\Pi^*) \\
        &\leq \left(\mathcal{F}^{\dagger C}\otimes \I^B\right)(I^{C}\otimes I^B) \\
        &\overset{(a)}{=} I^{A}\otimes I^B.
    \end{align*}
    In equality $(a)$ we have used the fact that $\mathcal{F}^{\dagger}$ is unital, which implies that $\mathcal{F}^{\dagger}(I^C)=I^A$.    With these observations in hand, note that the following holds:
    \begin{align*}
       \Tr\left[\Pi^*\left(\mathcal{F}^A\otimes \I^B\right)(\rho^{AB})\right] &=  \ip{\Pi^*}{\left(\mathcal{F}^A\otimes \I^B\right)(\rho^{AB})} \\
       &= \ip{\left(\mathcal{F}^{\dagger C}\otimes \I^B\right)(\Pi^*)}{\rho^{AB}} \\
       &= \Tr\left[\left(\mathcal{F}^{\dagger C}\otimes \I^B\right)(\Pi^*)\rho^{AB}\right] \\
       &\geq 1-\eps.
    \end{align*}
    This implies that $\left(\mathcal{F}^{\dagger C}\otimes \I^B\right)(\Pi^*)$ is a candidate optimiser for $H_H^{\eps}(A|B)_{\rho}$, which in turn implies that:
    \begin{align*}
        \exp(H_H^{\eps}(A|B)_{\rho}) &\leq \Tr\left[\left(\mathcal{F}^{\dagger C}\otimes \I^B\right)(\Pi^*)\left(I^A\otimes \rho^B\right)\right] \\
        &= \ip{\left(\mathcal{F}^{\dagger C}\otimes \I^B\right)(\Pi^*)}{I^A\otimes \rho^B}\\
        &= \ip{\Pi^*}{\left(\mathcal{F}^{A}\otimes \I^B\right)\left(I^A\otimes \rho^B\right)}\\
        &\overset{(b)}{=} \ip{\Pi^*}{I^C\otimes \rho^B}\\
        &= \Tr\left[\Pi^*\left(I^C\otimes \rho^B\right)\right]\\
        &= \exp\left(H_H^{\eps}(C|B)_{\left(\mathcal{F}^A\otimes \I^B\right)(\rho)}\right),
    \end{align*}
    where in equality $(b)$ we have used the fact that $\mathcal{F}$ is unital.
\end{proof}

\section{Definitions: $\eps$-Purity, Local and Distributed Purity Distillation}\label{sec:defpurity}
In this section we present the formal definitions of the $\eps$-purity of a state and the tasks of local and distributed purity distillation.

\subsection{$\eps$-Purity and Allowable Local Operations}
We will first define the $\eps$-purity of a state:
\begin{definition}{\bf $\eps$-Purity}
    Given a state $\rho^A$, the $\eps$-purity of $\rho^A$ is defined to as the number $\log \abs{A}-H_H^{\eps}(A)_{\rho}$.
\end{definition}

As mentioned in the introduction, the notion of $\eps$-purity puts a bound on the number of single qubit pure states $\ket{0}$ that may be extracted from a given state $\rho$. To make this connection precise, we have to list the kinds of local operations that a party is allowed to perform on $\rho$ to extract pure states from it. It is crucial that these operations do not increase the $\eps$-purity of the state. To that end, we consider below a list of allowed local operations. We later show in Lemma \ref{lem:puritynonincreasing} that indeed the operations listed below cannot increase the $\eps$-purity of a given state.

\begin{definition}{\bf Allowable Local Operations}\label{def:allowableoperations}
    Given a state $\rho^A$, we allow the following operations to be performed on the system $A$:
    \begin{enumerate}
        \item Appending a register $A_{\mix}$ to the system $A$, where the state on $A_{\mix}$ is maximally mixed.
        \item Unitary operations.
        \item Local completely dephasing maps $\mathcal{P}$.\item Tracing out a subsystem.
    \end{enumerate}
    Along with the above operations, we will also allow appending pure states $\ketbra{0}$ to the system $A$ in a register $C_{\pure}$. To account for this, we require that the formula for the $\eps$-purity of the state on $AC_{\pure}$ be modified as follows:
    \[
    \log \abs{AC_{\pure}}-H_H^{\eps}(AC_{\pure})_{\rho\otimes \ketbra{0}}-\log \abs{C_{\pure}}.
    \]
\end{definition}
\subsection{Local Purity Distillation}

We will now give an operational interpretation to the $\eps$-purity, by building protocols out of the allowable operations which extract pure states from the given input state. To do this, we first define the notion of a local purity distillation code:
\begin{definition}{\bf (Local Purity Distillation Code)}
Given a quantum state $\rho^A$ in the register $A$, we define a $\eps$ local purity distillation code as a sequence of allowable operations which produce a state $\sigma^{A_p}$, with the property that:
\[
\norm{\sigma^{A_p}-\ketbra{0}^{A_p}}_1 \leq \eps.
\]
The rate of the code is given by
\[
R_{\textup{local}}^{\eps} \coloneqq \log \abs{A_p} - \log \abs{C_{\pure}}.
\]
\end{definition}   
A rate $R$ is said to be $\eps$-achievable for local purity distillation with respect to the state $\rho^A$ if there exists an $\eps$ purity distillation code such that
\[
R^{\eps}_{\textup{local}}=R-O(\log \frac{1}{\eps})
\]

\begin{definition}{\bf ($\eps$-Local Distillable Purity)}
    Given a state $\rho^A$, the \emph{$\eps$-local distillable purity} $\kappa_{\eps}(\rho^A)$ is defined as the supremum over all $\eps$-achievable rates $R$ for local purity distillation.
\end{definition}

\subsection{Distributed Purity Distillation}

As mentioned in the introduction, the main topic of this paper is the task of distributed purity distillation. In this task we envision two parties, Alice and Bob, each of whom possess a share of a bipartite quantum state $\rho^{AB}$. The goal is for them to coordinate and extract pure states from this shared state. Under the supposition that Alice and Bob are allowed only local allowable operations on their systems $A$ and $B$, they can each perform an optimal local purity distribution protocol, and recover pure states roughly at the rate $\log \abs{AB}-H_H^{\eps^2}(A)_{\rho}-H_H^{\eps^2}(B)_{\rho}$. However, note that in this setup since we did not allow any communication between Alice and Bob, this is the best that they can do. The question then is that whether given the ability to communicate, can they do better?

We must keep in mind that whatever communication channel we introduce must be implementable by composing some allowable operations. This naturally leads us to the following definition of a distributed purity distillation protocol:

\begin{definition}{\bf (Distributed Purity Distillation (DPD))}\label{def:dfd}
Given a bipartite quantum state $\rho^{AB}$ to two parties Alice and Bob, where Alice has access to the register $A$ and Bob has access to the register $B$. A distributed purity distillation protocol with error $\eps$ is then defined as a protocol consisting of :
\begin{enumerate}
    \item Local allowable operations on the system $A$.
    \item A completely dephasing channel $\mathcal{P}^{X_A\to X_B}$, where the system $X_A$ is generated at Alice's end and $X_B$ is a classical register belonging to Bob.
    \item Local allowable operations on the system $B$.
\end{enumerate}
Suppose that the state generated at the end of the protocol is $\sigma^{A_pB_p}$, where the system $A_p$ belongs to Alice and $B_p$ belongs to Bob. We require that:
\[
\norm{\sigma^{A_pB_p}-\ketbra{0}^{A_p}\otimes \ketbra{0}^{B_p}}_1 \leq \eps.
\]
\end{definition}

The rate of the protocol is defined as
\[
R_{\textup{dist}}^{\eps}\coloneqq \log \abs{A_p} +\log \abs{B_p} -\log \abs{C}
\]
where the system $C\cong C_{\alice}\otimes C_{\bob}$ accounts for the local pure ancilla qubits borrowed by both Alice and Bob in the registers $C_{\alice}$ and $C_{\bob}$ respectively.

In this paper, as in \cite{KroviDevetak}, we will be concerned with DPD protocols with bounded classical communication from Alice to Bob. To that end, we introduce the following definition:

\begin{definition}{\bf (DPD with Bounded Classical Communication)}\label{def:dfdbc}
Given a bipartite state $\rho^{AB}$, we define a $(R^{\eps}_{\textup{dist}}, C_{\textup{com}}^{\eps}, \eps)$ protocol as a distributed purity distillation protocol with error $\eps$, as defined in Definition \ref{def:dfd}, where it holds that
\[
\log \abs{X_B}\leq C_{\textup{com}}^{\eps}.
\]
where $X_B$ is the classical output register of the perfectly dephasing channel $\mathcal{P}^{X_A\to X_B}$.
\end{definition}
A rate pair $(R_{\textup{pure}},C_{\textup{classical}})$ is said to be $\eps$-achievable for DPD with bounded classical communication if there exists a $(R^{\eps}_{\textup{dist}}, C_{\textup{com}}^{\eps}, \eps)$ protocol such that:
\begin{align*}
    R^{\eps}_{\textup{dist}} &= R_{\textup{pure}}-O(\log \frac{1}{\eps})\\
   C_{\textup{com}}^{\eps}  &\leq C_{\textup{classical}}+O(\log \frac{1}{\eps}).
\end{align*}

\begin{definition}{\bf ($\eps$ $1$-way Distillable Purity)}
    Given a state $\rho^{AB}$ and $C_{\textup{classical}}\geq 0$, the $\eps$ $1$-way distillable purity $\kappa^{\to}_{\eps}(\rho^{AB}, C_{\textup{classical}})$ is defined as the supremum of $R_{\textup{pure}}$ over all $\eps$-achievable rates $(R_{\textup{pure}}, C_{\textup{classical}})$ for distributed purity distillation.
\end{definition}

\begin{remark}
    We will use the notation $\kappa_{\eps}^{\to}(\rho^{AB}, \infty)$ to indicate the $1$-way distillable purity in the case when we allow unbounded but finite classical communication.
\end{remark}

\section{Optimal Protocols for Local Purity Distillation}\label{sec:localprotocols}

In this section we will show that given a state $\rho^A$, any finite sequence of allowable operations cannot increase the $\eps$-purity of this state. We will then give an operational interpretation of the $\eps$-purity, by constructing a local purity distillation code which extracts pure states from the given state at a rate which is almost equal to the $\eps$-purity. We will also show that the $\eps$-purity is the best rate of pure state production which any local purity distillation code can hope to achieve.

\begin{lemma}\label{lem:puritynonincreasing}
    The $\eps$-purity of a state $\rho^A$ is non-increasing under allowable local operations.
\end{lemma}

\begin{proof}
     Recall that, given a state $\rho^A$, the following local operations are allowed:
    \begin{enumerate}
        \item Introducing a maximally mixed state in a register $A_{\mix}$.
        \item Introducing pure states in a register $C_{\pure}$ which must be accounted for.
        \item Unitary operations.
        \item A local completely dephasing channel $\mathcal{P}$.
        \item Discard (trace out) a subsystem.
    \end{enumerate}
    We will show that each of the above operations do not increase the $\eps$-purity of $\rho^A$ i.e., $\log \abs{A}-H_H^{\eps}(A)_{\rho}$.\\
    \vspace{1mm}
    
    \noindent {\bf Appending $\pi^{A_{\mix}}$}:\\ The state under consideration is now $\rho^A\otimes \pi^{A_{\mix}}$. Then, the following holds:
    \begin{align*}
        &\log \abs{AA_{\mix}}-H_H^{\eps}(AA_{\mix})_{\rho\otimes\pi} \\
        \overset{(a)}{=}~& \log \abs{AA_{\mix}}-H_H^{\eps}(A)_{\rho}-\log \abs{A_{\mix}}\\
       = ~&\log \abs{A}-H_H^{\eps}(A)_{\rho},
    \end{align*}
    where equality $(a)$ follows from Lemma \ref{lem:additivehh}.\\
    \vspace{1mm}

    \noindent{\bf Appending $C_{\textup{pure}}$}:\\ In this case, recall that Definition \ref{def:allowableoperations} requires the formula for the $\eps$-purity to be adjusted with a correction term $-\log \abs{C_{\textup{pure}}}$. With this correction and the fact that $H_H^{\eps}(AC_{\textup{pure}})_{\rho\otimes \ketbra{0}}=H_H^{\eps}(A)_{\rho}$ (see Lemma \ref{lem:puretensorsame}) it is trivial to see that the $\eps$-purity does not change.\\
    \vspace{1mm}

    \noindent{\bf Unitary Operations}:\\ In this case, suppose that a unitary operator $U^A$ acts on $\rho^A$ to give $\sigma^A$. The unitary invariance of $H_H^{\eps}(\cdot)$ implies that $H_H^{\eps}(A)_{\rho}=H_H^{\eps}(A)_{\sigma}$. This directly implies that unitary operations keep the $\eps$-purity invariant.\\
    \vspace{1mm}

    \noindent{\bf Completely Dephasing Maps}:\\ In this case, suppose that the system $A$ is comprised of the registers $A'X_1$, and there exists a completely dephasing map $\mathcal{P}^{X_1\to X_2}$, where $\abs{X_1}=\abs{X_2}$. Since there exists a natural isomorphism between $A$ and $A'X_1$, we can write the following:
    \begin{align*}
        H_H^{\eps}(A)_{\rho} &= H_H^{\eps}(A'X_1)_{\rho} \\
        &\overset{(b)}{\leq} H_H^{\eps}(A'X_2)_{\I^{A'}\otimes \mathcal{P}^{X_1}(\rho)}, 
    \end{align*}
    where we have used the fact that the map $\I^{A'}\otimes \mathcal{P}^{X_1}$ is a unital CPTP map and Lemma \ref{lem:dataprocessinghh} in step $(b)$. This directly implies that the $\eps$-purity is non-increasing under these maps, since $\log \abs{A}=\log \abs{A'X_1}$.\\
    \vspace{1mm}

    \noindent{\bf Discarding Subsystems}:\\ Again, suppose that $A$ is comprised of the systems $A''G$, where the system $G$ is to be discarded. Suppose that the state after discarding $G$, on the system $A''$ is $\sigma^{A''}$. Then, the following holds:
    \begin{align*}
        H_H^{\eps}(A)_{\rho} &= H_H^{\eps}(A''G)_{\rho}\\
        &\overset{(c)}{\leq } H_H^{\eps}(A'')_{\sigma}+\log \abs{G}.
    \end{align*}
    Step $(c)$ follows from Corollary \ref{corol:upperboundhhdim}. This implies that:
    \begin{align*}
        &\log \abs{A}-H_H^{\eps}(A)_{\rho}\\
        =~&\log \abs{A''G}-H_H^{\eps}(A''G)_{\rho} \\
        \geq ~&\log \abs{A''G}-H_H^{\eps}(A'')_{\sigma}-\log \abs{G}\\
        =~&\log \abs{A''}-H_H^{\eps}(A'')_{\sigma}.
    \end{align*}
This concludes the proof.
\end{proof}

We will now provide an operational interpretation of the $\eps$-purity, by exhibiting an $\eps$ purity distillation code which recovers pure states from the given input state at a rate almost equal to the $\eps$-purity of the input state. Such a protocol was shown to exist in \cite[Theorem 1.7]{CNB23}. Further using Fact \ref{lem:equivHHtildeHmax} along with the results of \cite{CNB23}, we get the following fact:

\begin{fact}{\bf {Lower Bound for $\kappa_{\eps}(\rho^A)$}}\label{lem:localLB}
Given a quantum state $\rho^A$, there exists an $\eps$ purity distillation code with rate
\[
R^{\eps}_{\textup{local}} = \log \abs{A} - H_H^{\eps^2/9}(A)_{\rho}+O(\log \eps)-1
\]
This also implies that:
\[
\kappa_{\eps}(\rho^A)\geq \log \abs{A} - H_H^{\eps^2/9}(A)_{\rho}+O(\log \eps)-1.
\]
In fact, the $\eps$ purity distillation code which achieves the above lower bound consists only of a unitary operator $U^A$ acting on the system $A$, and does not require any other allowable operations.
\end{fact}

The following theorem encapsulates out discussion so far and connects the local distillable purity with the $\eps$-purity by showing that the latter is an upper bound for the former:

\begin{theorem}\label{thm:local}
    Given a quantum state $\rho^A$ and $\eps>0$, the $\eps$ local distillable purity of the state $\kappa_{\eps}(\rho^A)$ satisfies the following bounds:
    \[
    \log \abs{A}-H_H^{\eps^2/9}(A)_{\rho}+O(\log\eps)-1\leq\kappa_{\eps}(\rho^A)\leq \log \abs{A}-H_H^{\eps}(A)_{\rho}
    \]
\end{theorem}

\begin{proof}
       The lower bound follows directly from Fact \ref{lem:localLB}. To get the upper bound, note that any $\eps$ purity distillation code is a sequence of allowable operations, which finally output a state $\sigma^{A_p}$, such that:
    \[
    \norm{\sigma^{A_p}-\ketbra{0}^{A_p}}_1\leq \eps.
    \]
    Using the allowed operations listed above, we will now characterise the form of any finite sequence of operations. To do this, we adopt the notation that any subsystem with the name $G_i$ (for some $i\in \mathbb{N}$) will be discarded at the end of the protocol. We also denote the final output state as $\sigma^{A_pG}$, where $A_p$ is to be retained and $G$ discarded. Note that without loss of generality we can assume that any registers which contain maximally mixed states or pure states can be introduced at the very beginning of the protocol, and any systems that are to be traced out can be discarded at the very end. Then, any general local purity distillation protocol takes the form in Table \ref{table:dummylocalprotocol}.
\begin{table}[h]
        \centering
            \begin{tabular}{  l c r }
            \hline \\
            \textbf{Alice} & \\
            \hline \\
            \hdashline\\
            & State $\rho^A$ &\\
            \hdashline\\
             Append $\pi^{A_{\mix}}\otimes \ketbra{0}^{C_{\pure}}$. \\
             \hdashline\\
              & State $\sigma^{AA_{\mix}C_{\pure}}\coloneqq \rho^A\otimes \pi^{A_{\mix}}\otimes \ketbra{0}^{C_{\pure}}$ \\
              \hdashline \\
              Unitary $U_1: {AA_{\mix}C_{\pure}}\to A_1X_1G_1$ & \\
              Channel $\mathcal{P} : X_1\to X_2G_2$ & \\
              Unitary $U_2: A_1X_2\to A_2 X_3G_3$ & \\
                {\bf $\ldots$} & \\
               \hdashline\\
               & State $\sigma^{A_pG}$ \\
               \hdashline\\
               Discard the subsystem $G$ & \\
               \hdashline\\
               & State $\sigma^{A_p}$\\
            \hline
\end{tabular}
        \caption{\bf General Schema of a Local Protocol} 
        \label{table:dummylocalprotocol}
    \end{table}

 Then, one can use Lemma \ref{lem:puritynonincreasing} at every step of the protocol iteratively, to see that:
\[
\log \abs{A}-H_H^{\eps}(A)_{\rho}\geq \log \abs{A_p}-H_H^{\eps}(A_p)_{\sigma}-\log \abs{C_{\textup{pure}}}.
\]
    However, using the requirement that $\sigma^{A_p}$ has to be close to the pure state $\ketbra{0}^{A_p}$ and invoking Lemma \ref{lem:hhneg}, we see that:
    \[
    \log \abs{A_p}-H_H^{\eps}(A_p)_{\sigma}\geq \log \abs{A_p}.
    \]
    Collating these arguments, we see that:
    \begin{align*}
        \kappa_{\eps}(\rho^A) &\leq \log \abs{A_p}-\log \abs{C_{\pure}} \\
        &\leq \log \abs{A_p}-H_H^{\eps}(A_p)_{\sigma}-\log \abs{C_{\pure}}\\
        &\leq \log \abs{A}-H_H^{\eps}(A)_{\rho}.
    \end{align*}
    This concludes the proof.
\end{proof}

In the following section, we will thus refer to the \emph{locally optimal} protocol, in reference to Theorem \ref{thm:local}. Note that this locally optimal protocol consists only of a a unitary operator $U^A$ acting on the system $A$, as given in Fact \ref{lem:localLB}.

\section{Distributed Protocols with Ancilla: Upper Bounds}\label{sec:UBdistributed}

In this section we prove a one-shot upper bound on the number of qubit states that Alice and Bob can hope to distil, given the setting of the distributed purity distillation problem with classical communication bounded by the rate $C_{\textup{classical}}$. Throughout the rest of this section, to impose the bound on classical communication, we make the following assumption:
\begin{assumption}\label{asmptn:size}
    The completely dephasing channel $\mathcal{P}^{X_A\to X_B}$ is such that
    \[
    \log \abs{X_B}\leq C_{\textup{classical}},
    \]
    where $C_{\textup{classical}}$ is the maximum allowable rate of classical communication.
\end{assumption}

\subsection*{Notation}

In our proofs of the upper bounds for distributed purity distillation, we will have to deal with several entropic quantities related to states which exist at different times during the protocol. For example, we may use a relation of the form $H_H^{\eps}(A_pA_gX_A)\geq H_H^{\eps}(AC_{\alice}A_{\mix})$. These two entropic quantities correspond to two different states, related by Alice's application of her local operations. In the interest of brevity, we will not explicitly spell out the state corresponding to which these registers are defined. However, in all cases the state and the point in the protocol when that state exists will be clear from the context provided by Table \ref{table:dummyprotocol}.\\
\vspace{1mm}

\noindent We will also require the following lemma:

\begin{lemma}\label{lem:upperbounddist}
Given a quantum state $\rho^{AB}$ with the $A$ register belonging to Alice and the $B$ register belonging to Bob, any distributed purity distillation protocol making error at most $\eps$ can achieve a rate at most
\[
R^{\eps}_{\textup{dist}}\leq \log \abs{A}+\log \abs{B}-H_{\max}^{g(\eps)}(A)-H_{\min}^{f(\eps)}(B~|~X_B)+2\log \eps,
\] 
where the entropic quantities are computed with respect to states as defined in Table \ref{table:dummyprotocol}.
\end{lemma}

\begin{proof}
Before we start the proof, we will first characterise what any general distributed purity distillation protocol looks like:

\begin{table}[h]
        \centering
            \begin{tabular}{  l c r }
            \hline \\
            \textbf{Alice} & & \textbf{Bob} \\
            \hline \\
            \hdashline\\
& State $\rho^{AB}$ & \\
\hdashline\\
Append state $\pi^{A_{\mix}}\otimes \ketbra{0}^{C_{\alice}}$ & & Append state $\pi^{B_{\mix}}\otimes \ketbra{0}^{C_{\bob}}$\\
& & \hspace{1em}  \\
Allowable local operations & & \\
& & \hspace{1em}  \\
Create state $\sigma^{A_1X_A}$ & & \\
\hdashline\\
& State $\sigma^{A_1X_AB}$ & \\
\hdashline\\
& & \hspace{1em}  \\
& $\xrightarrow{X_A\to X_B}$ &\\
            Allowable local operations on $A_1$ & & Allowable local operations $BB_{\mix}C_{\bob}X_B$.\\
            & & \hspace{1em}  \\
            \hdashline\\
            & State $\sigma^{A_pA_gB_pB_g}$ & \\
            \hdashline\\
            Discard system $A_g$ & & Discard system $B_g$\\
            \hdashline\\
            & Final state $\sigma^{A_pB_p}$\\
            \hline
\end{tabular}
        \caption{\bf General Schema of a Distributed Protocol} 
        \label{table:dummyprotocol}
    \end{table}
\pagebreak
Note that in the general protocol, although we can roll all of Bob's actions together, we must treat Alice's actions before and after she sends the classical messages to Bob separately. As in the proof of Theorem \ref{thm:local}, we can assume without loss of generality that all systems that contain either maximally mixed states or pure states can be appended at the very beginning of the protocol and all systems to be traced out can be discarded at the end of the protocol. To that end, we make the convention that the actual state before both Alice and Bob discard some sub-systems is given by $\sigma^{A_pA_gB_pB_g}$, where $A_g$ and $B_g$ contain all systems that are to be discarded. Note that this means that the expression `Allowable local operations' during the protocol execution refers only to some finite sequence of local unitary operators and local completely dephasing maps.

Before we move on with the main proof, we will state a useful claim:

\begin{claim}\label{claim:inequalities}
    In reference to the protocol in Table \ref{table:dummyprotocol}, it holds, for any $\delta>0$, that:
    \begin{align*}
        &H_H^{\delta}(B_pB_g)\geq H_H^{\delta}(BB_{\mix}C_{\bob}X_B)\\
        & H_{\max}^{\delta}(A_pA_g)\geq H_{\max}^{\delta}(A_1)\\
        & H_{\max}^{\delta}(A_1X_A)\geq H_{\max}^{\delta}(AA_{\mix}C_{\alice}). 
    \end{align*}
\end{claim}
\begin{proof}
    Note that in going from $\sigma^{BB_{\mix}C_{\bob}X_B}$ to $\sigma^{B_pB_g}$, Bob uses either local unitary operators or completely dephasing maps. Since $H_H^{\delta}(\cdot)$ for any state is non-decreasing under these operations, the claim first inequality follows. Similar observations holds for Alice's actions in going from $\sigma^{A_1}$ to $\sigma^{A_pA_g}$, and from $\rho^{AA_{\mix}C_{\alice}}$ to $\sigma^{A_1X_A}$. Since the smooth max entropy is invariant under the action of isometries and non-increasing under the action of unital CPTP maps (Fact \ref{fact:unitalhmax}), the other two inequalities follow.
\end{proof}
 Next, note that:
\[
\log \abs{A_pB_p} - \log \abs{C_{\alice}C_{\bob}} =  \log \abs{ABA_{\mix}B_{\mix}}-\log \abs{A_gB_g}
\]
We will lower bound $\log \abs{A_gB_g}$ which will in turn allow us to upper bound $\log \abs{A_pB_p}-\log \abs{C_{\alice}C_{\bob}}$. Before we begin, we would like to point out that the  systems $X_A$ and $X_B$ are isomorphic, however, they differ in the fact that the system $X_B$ holds a classical state (diagonalisable with respect to the basis $\{\ket{x}\}$ of the completely dephasing channel) which is the output of the completely dephasing channel upon acting on the contents of the system $X_A$. Thus the state on the registers $X_BBB_{\mix}$ after Alice sends the contents of the register $X_A$ through the channel is a cq state (with pure qubits $\ket{0}^{C_{\bob}}$ in the register $C_{\bob}$ in tensor with the rest of the systems), while the state on the systems $A_pA_gX_A$ are \emph{not} cq in general. 

We will now lower bound $\log\abs{A_gB_g}$:
\begin{align*}
    \log\abs{A_gB_g} & \geq H_H^{\eps}(A_g) + H_H^{\eps}(B_g)+2\log \eps \\
                     & \geq H_H^{3\sqrt{\eps}}(A_pA_g) + H_H^{3\sqrt{\eps}}(B_pB_g)+2\log \eps \\
                     \intertext{The above inequality uses the subadditivity of the smooth hypothesis testing entropy twice, along with the fact that both $H_H^{\eps}(A_p)$ and $H_H^{\eps}(B_p)$ are 0. Thus, LHS is}
                     & \overset{(a)}{\geq} H_H^{3\sqrt{\eps}}(A_pA_g) + H_H^{3\sqrt{\eps}}(BX_BB_{\mix}C_{\bob})+2\log \eps \\
                     & \overset{(b)}{=} H_H^{3\sqrt{\eps}}(A_pA_g) + H_H^{3\sqrt{\eps}}(BX_B)+\log \abs{B_{\mix}}+2\log \eps \\
                     & \overset{(c)}{\geq} H_{\max}^{3\sqrt{\eps}}(A_pA_g)+ H_{\max}^{3\sqrt{\eps}}(X_B) + H_{\max}^{3\sqrt{\eps}}(BX_B)-H_{\max}^{3\sqrt{\eps}}(X_B)+\log \abs{B_{\mix}}+O(\log\eps) \\
                     & \overset{(d)}{\geq} H_{\max}^{3\sqrt{\eps}}(A_pA_g)+ H_{\max}^{3\sqrt{\eps}}(X_A) + H_{\min}^{f(\eps)}(B~|~X_B)+\log \abs{B_{\mix}}+O(\log\eps) \\
                     & \overset{(e)}{\geq} H_{\max}^{3\sqrt{\eps}}(A_1)+ H_{\max}^{3\sqrt{\eps}}(X_A) + H_{\min}^{f(\eps)}(B~|~X_B)+\log \abs{B_{\mix}}+O(\log\eps)
                     \intertext{In inequality $(a)$ we have used Claim \ref{claim:inequalities}. In equality $(b)$ we used Lemma \ref{lem:additivehh} and also the fact that the register $C_{\bob}$ contains a pure state in tensor with all the other systems. In inequality $(c)$ we have used Fact \ref{fact:tildemaxHrelation} and Lemma \ref{lem:equivHHtildeHmax} to lower bound both $H_H^{3\sqrt{\eps}}$ terms by $H_{\max}^{3\sqrt{\eps}}$, and we have absorbed the constant $-1$ arising from Lemma \ref{lem:equivHHtildeHmax} into the $O(\log \eps)$ term, assuming small enough $\eps$. In inequality $(d)$ above we have used the fact that the completely dephasing channel is a unital CPTP and the smooth max entropy cannot be decreased by the action of such a map \cite{Tomamichel_thesis}. We have also used the chain rules from Fact \ref{fact:dupuischainrules}. Inequality $(e)$ follows from Claim \ref{claim:inequalities}. Next, we will use the subadditivity of the max entropy to see that the LHS is:}
                     & \geq H_{\max}^{h(\eps)}(A_1X_A)+ H_{\min}^{f(\eps)}(B~|~X_B)+\log \abs{B_{\mix}}+O(\log\eps) \\
                     &\overset{(f)}{\geq} H_{\max}^{h(\eps)}(AC_{\alice}A_{\mix})+ H_{\min}^{f(\eps)}(B~|~X_B)+\log \abs{B_{\mix}}+O(\log\eps) \\
                     &\overset{(g)}{\geq} H_{\max}^{g(\eps)}(A)+ H_{\min}^{g(\eps)}(C_{\alice}A_{\mix}|A)+ H_{\min}^{f(\eps)}(B~|~X_B)+\log \abs{A_{\mix}}+\log \abs{B_{\mix}}+O(\log\eps)\\
                     &\overset{(h)}{=} H_{\max}^{g(\eps)}(A)+ H_{\min}^{f(\eps)}(B~|~X_B)+\log \abs{A_{\mix}}+\log \abs{B_{\mix}}+O(\log\eps) .
\end{align*}
We have used Claim \ref{claim:inequalities} in inequality $(f)$. In inequality $(g)$ we have used the chain rule from Fact \ref{fact:dupuischainrules}. Finally, for equality $(h)$, we use the following observation, which holds for any $\delta>0$:
\begin{align*}
    H_{\min}^{\delta}(C_{\alice}A_{\mix}|A) &\geq H_{\min}(C_{\alice}A_{\mix}|A) \\
    &\geq \max\limits_{\sigma^A}\sup\log\brak{\lambda~|~2^{-\lambda} \I^{C_{\alice}A_{\mix}}\otimes \sigma^A\geq \pi^{A_{\mix}}\otimes \ketbra{0}^{C_{\alice}}\otimes \rho^A} \\
    &\geq \sup\log \brak{\lambda~|~2^{-\lambda} \I^{C_{\alice}A_{\mix}}\otimes \rho^A\geq \pi^{A_{\mix}}\otimes \ketbra{0}^{C_{\alice}}\otimes \rho^A} \\
    &\geq \log \abs{A_{\mix}}.
\end{align*}
This shows that, for any distributed purity distillation protocol with error at most $\eps$, it holds that
\[
R^{\eps}_{\textup{dist}} \leq  \log \abs{A}+\log \abs{B}-H_{\max}^{g(\eps)}(A)-H_{\min}^{f(\eps)}(B~|~X_B)+O(\log \eps)
\]
This concludes the proof.
\end{proof}

We are now ready to state and prove a theorem about the upper bound of the distributed purity of any quantum state:
\begin{theorem}{\bf (Upper Bound for Distributed Purity of a State)}\label{thm:UBgeneral}
Given a quantum state $\rho^{AB}$, the $1$-way distillable purity $\kappa_{\eps}^{\to}(\rho^{AB}, C_{\textup{classical}})$ is at most
\[
\kappa_{\eps}^{\to}(\rho^{AB}, C_{\textup{classical}}) \leq \log \abs{A} + \log \abs{B} - H_{\max}^{g(\eps)}(A) - \min\limits_{\Lambda\in \mathcal{S}} H_{\min}^{f(\eps)}(B ~|~X)_{\Lambda^A\otimes \I^B (\rho^{AB})}+O(\log \eps)
\]

where the set $\mathcal{S}$ is a subset of the set of all POVMs on the system $A$ and is defined as follows:
\[
\mathcal{S}\coloneqq \brak{\Lambda^{A\to X}~|~I_{\max}^{\eps}(X:RB)_{\I^{RB}\otimes \Lambda(\ketbra{\rho}^{ABR})}+O(\log\eps)\leq C_{\textup{classical}}}
\]
where $\ket{\rho}^{ABR}$ is an arbitrary purification of $\rho^{AB}$.
\end{theorem}

\begin{proof}
From Lemma \ref{lem:upperbounddist}, we know that any distributed purity distillation protocol for $\rho^{AB}$ and which makes an error at most $\eps$, can extract a purity of at most
\begin{align}\label{eq:sup}
R^{\eps}_{\textup{dist}} \leq  \log \abs{A}+\log \abs{B}-H_{\max}^{g(\eps)}(A)-H_{\min}^{f(\eps)}(B~|~X_B)+O(\log \eps)
\end{align}
Recall that we obtained the system $X_B$ by:
\begin{enumerate}
    \item Using local allowed operations to obtain $\sigma^{A_1X_A}$.
    \item Sending $X_A$ through the completely dephasing channel $\mathcal{P}^{X_A\to X_B}$.
    \item Using local allowed operations on $A_1$ to obtain the systems $A_pA_g$.
\end{enumerate}
Suppose that $V_1^{AA_{\mix}C_{\pure}\to A_1X_AE_1}$ and $V_2^{A_1\to A_pA_gE_2}$ are the Stinespring dilations of the maps that Alice enacts in Steps $1$ and $3$ above. Now consider the isometry:
\begin{align*}
    V_3 &: AA_{\mix}C_{\pure}\to A_pA_gX_AE_1E_2\\
    &\coloneqq \left(V_2\otimes \I^{X_A}\right)\circ V_1.
\end{align*}
Define the unitary extension of $V_3$ in the usual way by appending an ancilla system $W$ to the domain of $W_3$. We will call this unitary $U_{\textsc{alice}}$. Then note that if we measure the $X_A$ system in the computation basis, the probability of getting the outcome $x$ is given by the following expression:

\begin{align*}
 \Tr\left[\left(\I^{A_pA_gE_1E_2}\otimes \ketbra{x}^{X_A}\right)\cdot~ \left(U_{\textsc{alice}}  (\rho^A\otimes \pi^{A_{\mix}}\otimes \ketbra{0}^{C_{\alice}}\otimes \ketbra{0}^W)U_{\textsc{alice}}^{\dagger}\right)\right]
\end{align*}
By defining $\tau^{A_{\mix}C_{\alice}W}\coloneqq \pi^{A_{\mix}}\otimes \ketbra{0}^{C_{\alice}}\otimes \ketbra{0}^{W}$, and using the cyclicity of trace, we see that the above expression simplifies to:
\[
\begin{aligned}
&\Tr_A\left[\Tr_{A_{\mix}C_{\alice}W}\left(\left(I^{A}\otimes \sqrt{\tau}^{A_{\mix}C_{\alice}W}\right)\cdot U_{\textsc{alice}}^{\dagger}\left(\I^{A_pA_gE_1E_2}\otimes \ketbra{x}^{X_A}\right)U_{\textsc{alice}}\right)(\rho^A)\right] 
\end{aligned}
\]

We define:
\[
\Lambda_x^{A}\coloneqq \Tr_{A_{\mix}C_{\alice}W}\left[\left(\left(I^{A}\otimes \sqrt{\tau}^{A_{\mix}C_{\alice}W}\right)\cdot U_{\textsc{alice}}^{\dagger}\left(\I^{A_pA_gE_1E_2}\otimes \ketbra{x}^{X_A}\right)U_{\textsc{alice}}\right)\right].
\]
Clearly $\Lambda_x\geq 0$. Additionally, it is easy to see that:
\[
\sum\limits_{x}\Lambda_x^A=I^A.
\]
Therefore, $\Lambda\coloneqq \brak{\Lambda_x^A}$ is a POVM. Note that the following holds, with respect to the state $(\I^{RB}\otimes \Lambda)(\ketbra{\rho}^{ABR})$:
\begin{align*}
    C_{\textup{classical}}\geq &~ \log \abs{X_B} \\
    \geq &~ H_{\max}^{O(\eps^{2})}(X_B)_{(\I^{RB}\otimes \Lambda)(\ketbra{\rho}^{ABR}}+O(\log\eps) \\
    \geq &~ I_{\max}^{\eps}(X_B:RB)_{(\I^{RB}\otimes \Lambda)(\ketbra{\rho}^{ABR}}+O(\log \eps).
\end{align*}

Thus, renaming $X_B$ to $X$, it is clear that the supremum of the rate $R_{\textup{pure}}$ over all $\eps$ achievable rates pairs $(R_{\textup{pure}},C_{\textup{classical}},\eps)$ is bounded above by the supremum of the upper bound on $R^{\eps}_{\textup{dist}}$ obtained in Equation \ref{eq:sup}, over the set $\mathcal{S}$ of all POVMS $\Lambda^{A\to X}$ such that $I_{\max}^{\eps}(X:RB)_{\I^{RB}\otimes \Lambda(\ketbra{\rho}^{ABR})}+O(\log\eps)\leq C_{\textup{classical}}$. This immediately implies that:
\[
\kappa_{\eps}^{\to}(\rho^{AB}, C_{\textup{classical}}) \leq \log \abs{A}+\log \abs{B}-H_{\max}^{g(\eps)}(A)-\inf\limits_{\Lambda\in \mathcal{S}} H_{\min}^{f(\eps)}(B ~|~X)_{\Lambda^A\otimes \I^B (\rho^{AB})}+O(\log \eps)
\]
This concludes the proof.
\end{proof}

\subsection{The Special Case of Unbounded Classical Communication}\label{sec:unbounded}
As mentioned earlier in the introduction, the original version of the distributed purity distillation problem was considered by Devetak in \cite{Devetak_purity}, in the regime when unbounded communication is allowed. In that spirit, we will show in this section that $\kappa_{\eps}^{\to}(\rho^{AB}, \infty)$ can be bounded above by the expression in Theorem \ref{thm:UBgeneral}, with the important distinction that the infimum over all POVMs in the set $\mathcal{S}$ can be replaced by an infimum over all rank-$1$ POVMs.
\begin{theorem}
    Given a quantum state $\rho^{AB}$, the $1$-way distillable purity in the case of unbounded communication, $\kappa_{\eps}^{\to}(\rho^{AB},\infty)$, can be bounded above by:
    \[
    \kappa_{\eps}^{\to}(\rho^{AB},\infty) \leq \log \abs{A}+\log \abs{B}-H_{\max}^{g(\eps)}(A)-\inf\limits_{\Lambda: \textup{rank-}1} H_{\min}^{f(\eps)}(B ~|~X)_{\Lambda^A\otimes \I^B (\rho^{AB})}+O(\log \eps).
    \]
\end{theorem}
\begin{proof}
    The proof follows easily from Theorem \ref{thm:UBgeneral}, by noticing that: 
    \[
    \inf\limits_{\Lambda\in \mathcal{S}} H_{\min}^{f(\eps)}(B ~|~X)_{\Lambda^A\otimes \I^B (\rho^{AB})}\geq \inf\limits_{\Lambda: \textup{rank-}1} H_{\min}^{f(\eps)}(B ~|~X)_{\Lambda^A\otimes \I^B (\rho^{AB})},
    \]
    by the data processing inequality for the smooth min entropy. Specifically, for any POVM $\Lambda$, one can always create a new POVM $\Lambda'$ by decomposing each POVM element in $\Lambda$ into rank one operators $\ketbra{\varphi}$ such that $0\leq \Tr\left[\ketbra{\varphi}\right]\leq 1$, and then assigning a new label to the outcome corresponding to each of these operators. This concludes the proof. 
\end{proof}
\section{Distributed Protocols with Ancilla: Lower Bounds}\label{sec:LBdistributed}
In this section we will present a DPD protocol with bounded classical communication, which uses additional ancilla qubits in a catalytic manner, with an almost optimal rate of pure state distillation. We call this protocol \kd (see Section \ref{sec:kd}). This protocol can be viewed as a one-shot version of the protocol presented in \cite{KroviDevetak}.

The main theorem in this section quantifies the rate of pure state distillation for the protocol \kd. The theorem will take the following form: we will first fix the state $\rho^{AB}$ and a rate of classical communication $C_{\textup{classical}}$. We will then fix a POVM $\Lambda^{A\to X}$ such that $I_{\max}^{\eps}(X:RB)_{\I^{RB}\otimes \Lambda(\ketbra{\rho}^{ABR})}+O(\log\eps)\leq C_{\textup{classical}}$. Finally, we will show the existence of a protocol which distils pure states at the rate (roughly) $\log \abs{A}-H^{\eps}_H(A|X)+\log \abs{B}-H^{\eps}_H(B|X)-I_{\max}^{\eps}(RB:X)$, with communication $I_{\max}^{\eps}(RB:X)+O(\log \frac{1}{\eps})$ (which is at most $C_{\textup{classical}}+O(\log\frac{1}{\eps})$). This protocol is \kd. The final lower bound on $\kappa^{\to}_{\eps}(\rho^{AB}, C_{\textup{classical}})$ is given by taking the supremum over all POVMs $\Lambda$ such that $I_{\max}^{\eps}(X:RB)_{\I^{RB}\otimes \Lambda(\ketbra{\rho}^{ABR})}+O(\log\eps)\leq C_{\textup{classical}}$.

To be precise, we prove the following theorem:
\begin{theorem}
    Given the bipartite quantum state $\rho^{AB}$, it holds that
    \[
    \kappa^{\to}_{\eps^{1/32}}(\rho^{AB}, C_{\textup{classical}})\geq \sup\limits_{\Lambda\in \mathcal{S}}R_{\textup{pure}}+O(\log\eps),
    \]
    where
    \[
    R_{\textup{pure}}=\log\abs{A}-H_{H}^{\eps}(A~|~X)+\log \abs{B}-H_H^{\eps}(B~|~X)-I_{\max}^{\eps^4}(X:RB)
    \]
    and the set $\mathcal{S}$ is as follows:
    \[
    \mathcal{S}\coloneqq \brak{\Lambda^{A\to X}~|~I_{\max}^{\eps}(X:RB)_{\I^{RB}\otimes \Lambda(\ketbra{\rho}^{ABR})}+O(\log\eps)\leq C_{\textup{classical}}}.
    \]
    The rate of communication of the protocol is $I_{\max}^{\eps^4}(X:RB)+O(\log \frac{1}{\eps})$. All the entropic quantities above are computed with respect to the state $(\I^{RB}\otimes \Lambda) (\ketbra{\rho}^{ARB})$.
\end{theorem}
\begin{proof}
    First, fix $\Lambda\in \mathcal{S}$. Then, using this POVM in conjunction with Proposition \ref{prop:secondProtocol} and Lemma \ref{lem:derandK}, shows that there exists a protocol (\kd) which for which the rate $R_{\textup{pure}}$ is achievable with error at most $\eps^{1/32}$, and classical communication $I_{\max}^{\eps^4}(X:RB)+O(\log \frac{1}{\eps})$. Then, taking the supremum of $R_{\textup{pure}}$ over $\mathcal{S}$ concludes the proof.
\end{proof}

In the following sections, we will prove Proposition \ref{prop:secondProtocol} and Lemma \ref{lem:derandK}. To state and prove these claims, we will assume that a POVM $\Lambda$ is already provided, and no further mention of the set $\mathcal{S}$ will be made. Before we present the actual protocol, we will first start with a bad protocol, which distils a small number of qubits and needs unbounded communication. This protocol, although bad, will serve towards building intuition. We present this in Section \ref{sec:basicProtocol}. The full description of \kd can be found in the Sections \ref{sec:optWithAncilla} and \ref{sec:kd}. 

\begin{remark}\label{remark:noextraregister}
    An important feature of our 'achievable' protocols will be that we will not require the use of local randomness in the form of $\pi^{A_{\mix}}$ and $\pi^{B_{\mix}}$, neither will we need Bob to borrow ancilla qubits. We will also not need Alice or Bob to use \emph{local} completely dephasing channels, although they will of course need access to the dephasing channel which sends messages from Alice to Bob. Nevertheless, we will show that our achievable protocols will be almost optimal, in the sense that they will be able to recover pure qubit states \emph{almost} at the optimal rate given in Theorem \ref{thm:UBgeneral}.
    
    Thus, in all that follows, the registers $A_{\mix}, B_{\mix}$ and $C_{\bob}$ will be omitted. In the interest of brevity we thus abbreviate the register $C_{\alice}$ to just $C$. 
\end{remark}

\subsection{The Need for Measurement Compression}\label{sec:basicProtocol}
In this section we will introduce a `bad' protocol for distributed purity distillation which is not optimal with respect to the number of pure qubit states that it distils, but nevertheless helps in understanding some of the key ideas that lead to the other optimal protocol construction that follow in later sections. For the purposes of this demonstration we will not put a bound on classical communication.

We remind the reader of Fact \ref{lem:localLB}, restated here for convenience, which we shall use throughout the rest of the paper:
 \begin{fact}\label{fact:reminder}
Given a quantum state $\rho^A$, there exists an $\eps$ purity distillation code, which takes the form a unitary operator $U^A$, which is almost optimal.
 \end{fact}

To setup the protocol, we recall the setup of the distributed purity distillation problem, modified suitably according to the statements made in Remark \ref{remark:noextraregister}:
\begin{enumerate}
\item Alice and Bob share the state $\rho^{AB}$ at the beginning of the protocol, where Alice has access to the system $A$ and Bob has access to the system $B$. Alice is also given the POVM $\brak{\Lambda_x^{A\to X}}$ which has outcomes $x$ from the set of symbols $\mathcal{X}$. 
    \item Alice can borrow any number of qubits $\ket{0}$ as ancilla, but has to account for them at the end of the protocol. For example, Alice can choose to act the POVM $\Lambda$ on the system $A$, but she has to do this \emph{coherently} by borrowing $\log \abs{\mathcal{X}}$ number of qubits.
    \item Suppose Alice borrows the ancilla qubits in the system $C$. Then she is allowed to perform any local unitary of the following form:
    \[
    U_{\textsc{Alice}} : AC\to A_pA_gX_A .
    \]
    The system $A_p$ is meant to hold the pure states that Alice distils on her end. Note that as per Remark \ref{remark:noextraregister}, our protocol will not require Alice to have access to private randomness or local completely dephasing maps.
    \item Alice and Bob share a completely dephasing channel, i.e. a CPTP map $\mathcal{P} : X_A\to X_B$ where the systems $X_A$ and $X_B$ are isomorphic. The action of the map is described with respect to a fixed basis $\brak{\ket{x}^{X_A}}$:
    \[
    \mathcal{P}\left(\rho^{X_A}\right) = \sum_x \braket{x|\rho|x} \ketbra{x}^{X_B}
    \]
    The choice of basis can be fixed by Alice and Bob before the protocol starts.
    \item Bob is allowed to use local unitaries on the systems in his possession, i.e., he is allowed to use unitaries of the following sort:
    \[
    U_{\textsc{Bob}} : BX_B \to B_pB_g
    \]
    where the system $B_p$ is meant to hold the pure states that he distils. Note that as per Remark \ref{remark:noextraregister}, our protocol will not require Bob to have access to local pure states, private randomness or local completely dephasing maps.
\end{enumerate}
We require that at the end of the protocol, the state $\sigma^{A_pB_p}$ should satisfy the following constraint:
\[
\norm{\sigma^{A_pB_p}-\ketbra{0}^A_p\otimes \ketbra{0}^{B_p}}_1 \leq \eps
\]
For the purposes of this section we will assume that $X_A\cong X_B\cong X$. The protocol itself is given in Table \ref{prot:ProtocolA}.\\
\vspace{1mm}\\

\begin{table}[h]
        \centering
            \begin{tabular}{  l c r }
            \hline \\
            \textbf{Alice} & & \textbf{Bob} \\
            \hline \\
            Borrow $\log \abs{\mathcal{X}}$ qubits as pure ancilla. & & \\
            \hspace{1em} & & \\
            Act the POVM $\Lambda^{A\to X}$ coherently on $\ket{\rho}^{ABR}$ & & \\
            \hspace{1em} & & \\
            State : $\ket{\rho}^{X_AABR}\coloneqq \sum\limits_{x\in \mathcal{X}}\ket{x}^{X_A}\sqrt{\Lambda_x^A}\ket{\rho}^{ABR}$ & & \\
            \hspace{1em} & & \\
            Define:  $\ket{\rho_x}^{ABR}\coloneqq \frac{1}{\sqrt{\Tr\left[\Lambda_x^A\rho^{A}\right]}}\sqrt{\Lambda_x}\ket{\rho}^{ABR}$ & &\\
            \hspace{1em} & & \\
            **Define: $U_x^{A\to A_pA_g}$ be the locally optimal local \\ purity distillation protocol for $\rho_x^A$ & &\\
            \hspace{1em} & & \\
            Act unitary $    \sum\limits_{x\in \mathcal{X}}\ketbra{x}^{X_A}\otimes U_x^{A\to A_pA_g}$ \\
            on state $\ket{\rho}^{X_AABR}$ & & \\
            & $\xrightarrow{X_A\to X_B}$ &
            \\
            & End of  & \\ 
            & Alice's Actions & \\
            \hdashline\\
            & & State on $X_BB$: \\ & & $\rho^{X_BB}\coloneqq 
\sum\limits_{x}P_X(x)\ketbra{x}^{X_B}\otimes \rho_x^B$ \\ & & \hspace{1em}  \\ & &**Define: $V_x^{B\to B_pB_g}$ be the locally optimal local \\ & &purity distillation protocol for $\rho_x^B$\\ & & \hspace{1em}  \\ & & Act unitary $\sum\limits_{x\in \mathcal{X}}\ketbra{x}^{X_B}\otimes V_x^{B\to B_pB_g}$ \\
& & on state $\rho^{X_BB}$\\
            \hline
\end{tabular}
        \caption{\bf Protocol A}
        \label{prot:ProtocolA} 
    \end{table}
 **Here (see Table \ref{prot:ProtocolA}) the unitary operators $U_x$ and $V_x$ are given by Fact \ref{fact:reminder}. Let us analyse Protocol A. We claim the following proposition: \begin{proposition}\label{prop:firstProtocol}
     Protocol A produces 
     \[
     \log \abs{A}-H_H^{\eps^2}(A|X)+\log \abs{B}-H_H^{\eps^2}(B|X)-\log \abs{\mathcal{X}}-O\left(\log \frac{1}{\eps}\right)
     \]
     number of pure qubit states.
 \end{proposition}

 \begin{proof}
     First, fix an $x\in \mathcal{X}$, and consider the Schmidt decomposition of the state $\ket{\rho_x}^{ABR}$:
     \[
     \ket{\rho_x}^{ABR}=\sum_{s} \lambda_s \ket{s}^{A}\ket{s}^{BR}.
     \]
Consider the set of the smallest $\lambda_s$ whose squares add up to at most $\eps$. Let us call this set \textsc{bad}. The action of $U_x$ is to relabel those $\ket{s}^A$ which have a corresponding $\lambda_s$ which is not in \textsc{bad}:
\[
U_x : \ket{s}^A\to \ket{s}^{A_g}\ket{0}^{A_p}~~\forall \ket{s}~~\textup{ such that } \lambda_s\notin \textsc{bad}.
\]
The vector $\ket{s}^{A_g}$ is simply a low dimensional embedding of $\ket{s}^{A}$ into the system $A_g$, which has dimension at least $2^{\widetilde{H}_{H}^{\eps}(A)-1}$. This embedding preserves the pairwise inner products between the vectors, i.e., for all $s, s'$ such that $\lambda_s, \lambda_{s'}\notin \textsc{bad}$:
\[
\braket{s|s'}^{A_g}=\braket{s|s'}^{A}.
\]
We can then write:
\[
U_x^{A\to A_pA_g}\ket{\rho_x}^{ARB}= \sum\limits_{s : \lambda_s\notin \textsc{bad}}\lambda_s \ket{0}^{A_p}\ket{s}^{A_g}\ket{s}^{RB}+\ket{\textsc{junk}}^{A_pA_gRB}.
\]
It is then not hard to see that
\[
\norm{U_x\cdot \rho^{ARB}_x-\ketbra{0}^{A_p}\otimes \sum\limits_{\substack{s,s' \\ \lambda_s, \lambda_{s'}\notin \textsc{bad} }}\lambda_s\lambda_{s'}\ket{s}\bra{s'}^{A_g}\otimes \ket{s}\bra{s'}^{RB}}_1 \leq O(\sqrt{\eps}).
\]
Tracing out the system $A_g$ and noting that the substate $\ketbra{0}^{A_p}\otimes \sum\limits_{s \lambda_s\notin \textsc{bad} }\lambda_s \ket{s}\bra{s}^{RB}$ is $\eps$ close to $\rho_x^{RB}$,  one can see that:
\[
\norm{\Tr_{A_g}\left[U_x\cdot \rho^{ARB}_x\right]-\ketbra{0}^{A_p}\otimes \rho_x^{RB}}_1 \leq O(\sqrt{\eps}).
\]
Now, consider the cq state:
\[
\sum\limits_{x} P_X(x)\ketbra{x}^{X_A}\otimes \rho_x^{A},
\]
where $P_X(x)$ is the probability of the outcome $x$ when $\rho^{A}$ is measured with the POVM $\Lambda$. Then, using Lemma \ref{lem:averageToWorstcaseHH} we see that there exists a subset of $x$'s, which we call $\mathcal{S}_{\textsc{alice}}$ such that 
\begin{align*}
    &\Pr\limits_{P_X}\left[\mathcal{S}_{\textsc{alice}}\right]\geq 1-2\sqrt{\eps}\\
    &H_H^{\eps}(\rho_x^A)\leq H_H^{\eps^2}(A|X)-\log \eps ~~ \forall x\in \mathcal{S}_{\textsc{alice}}.
\end{align*}
Collating the arguments above, one can then see that the state on the system $X_BRB$ after Alice sends the system $X_A$ through the dephasing channel satisfies the following property:
\[
\norm{\sum\limits_{x}P_X(x)\ketbra{x}^{X_B}\otimes \Tr_{A_g}\left[U_x\cdot \rho^{ARB}_x\right]- \ketbra{0}^{A_p}\otimes \left(\sum\limits_{x}P_X(x)\ketbra{x}^{X_B}\otimes \rho_x^{RB}\right)}_1 \leq O(\sqrt{\eps}),
\]
where the system $A_p$ is constituted by $H_H^{\eps^2}(A|X)-\log \eps$ qubits. Another applpication of Lemma \ref{lem:averageToWorstcaseHH} shows us that there exists a set $\mathcal{S}_{\textsc{bob}}$ such that
\begin{align*}
    &\Pr\limits_{P_X}\left[\mathcal{S}_{\textsc{bob}}\right]\geq 1-2\sqrt{\eps}\\
    &H_H^{\eps}(\rho_x^B)\leq H_H^{\eps^2}(B|X)-\log \eps ~~ \forall x\in \mathcal{S}_{\textsc{bob}}.
\end{align*}
where the entropic quantities in the expression above are computed with respect to the cq state $\sum\limits_{x}P_X(x)\ketbra{x}^{X_B}\otimes \rho_x^B$. Therefore, using arguments that are similar to those we used in the case of Alice, we see that after Bob's actions and discarding the system $B_g$, the global state is $O(\sqrt{\eps})$ close to pure states on the $A_p$ and $B_p$ system, where:
\[
\log \abs{A_pB_p}\geq \log \abs{AB}-H_H^{\eps^2}(A|X)-H_H^{\eps^2}(B|X)+O(\log \eps).
\]
Recall however that we now have to adjust for the fact that Alice had borrowed $\log \abs{\mathcal{X}}$ qubits. Therefore, the net number of pure qubits distills is:
\[
\log \abs{A_pB_p}-\log \abs{C}\geq \log \abs{AB}-H_H^{\eps^2}(A|X)-H_H^{\eps^2}(B|X)-\log \abs{\mathcal{X}}+O(\log \eps).
\]
This concludes the proof.
 \end{proof}

 As mentioned earlier, the number of pure qubit states that Protocol A distils is nowhere near optimal. The is of course due to the $-\log \abs{\mathcal{X}}$ term over which we have no control. To fix this issue, we need to replace the POVM $\Lambda$ with some other POVM $\Lambda'$ which has far fewer number of outcomes, yet still allows Alice and Bob to distil $\log \abs{A}-H_H^{\eps}(A|X)$ and $\log \abs{B}-H_H^{\eps}(B|X)$ pure qubit states. This is exactly what the measurement compression theorem allows us to do, as we explain in the next section.

\subsection{Measurement Compression}
As mentioned in the last section, the we need to replace the POVM $\Lambda$ with a POVM $\Lambda'$ which has a much smaller number of outcome, in order to increase the number of pure qubit states that Protocol A distils. However, an issue with this strategy is that this new POVM may not allow Alice and Bob to individually distil $\log \abs{A}-H_H^{\eps}(A|X)$ and $\log \abs{B}-H_H^{\eps}(B|X)$ pure qubits. The measurement compression theorem comes to our aid here. Thus in this section we take a small detour from our exposition to state the measurement compression theorem.
Suppose we are given a bipartite quantum state $\rho^{AB}$ and a POVM $\Lambda^{A\to X}$. To understand the action of this POVM on the state $\rho^{AB}$, consider a purification $\ket{\rho}^{ABR}$. It can be shown (see \cite{Wilde_etal_measurementcomp}) that the global state, after the action of the POVM on the system $A$, looks like
\[\label{eq:idealpostmeasurementstate}
\sum\limits_{x}P_X(x) \ketbra{x}^X\otimes \rho_x^{BR} \tag{1},
\]
where
\[
\rho_x^{BR}\coloneqq \frac{1}{\Tr\left[\Lambda_x\ketbra{\rho}^{ABR}\right]}\Tr_A\left[\Lambda_x\ketbra{\rho}^{ABR}\right].
\]
and $P_X(x)$ is the probability of the outcome $x$ when $\rho^A$ is measured using the POVM $\Lambda$. The goal of the measurement compression theorem is to replace $\Lambda$ by some other POVM $\Lambda'^{A\to Y}$ such that the support size of the distribution $P_Y$ induced by $\Lambda'$ is much smaller than that of $P_X$, yet the post measurement state $\sum\limits_{y}P_Y(y)\ketbra{y}^Y\otimes \rho_y^{BR}$ is close to the ideal post measurement state in Equation \ref{eq:idealpostmeasurementstate}. This of course may not be possible with a single POVM $\Lambda'$ (the distribution $P_X$ may not be compressible). However, the measurement compression theorem gets around this issue by using multiple POVMs, indexed by $k\in [K]$, each with a small number of outcomes. Which of these POVMs one chooses to actually do the measurement is decided by picking $k$ randomly. Let us refer to these `smaller' POVMs as $\Theta^A(k)=\brak{\Theta_1^A(k), \Theta_2^A(k),\ldots, \Theta_L^A(k) }$, where $L$ is the number of outcomes. The new measurement process can then be encapsulated as follows:\\
\vspace{1mm}

\noindent\fbox{\parbox{\textwidth}{
\begin{enumerate}
    \item Pick $k\overset{R}{\gets} [K]$.
    \item Measure the register $A$ of the state $\rho^{AB}$ using the smaller POVM $\Theta^A(k)$. \\Suppose this measurement produces an outcome $\ell\in [L]$.
    \item Map the symbol $(k,\ell)$ appropriately to an $x\in \mathcal{X}$ to recover the correct measurement outcome.
\end{enumerate}
}
}
\vspace{1mm}

Roughly, the measurement compression theorem says that, as long as $K$ and $L$ are large enough, the procedure above produces a post measurement state that is close to the ideal state in Equation \ref{eq:idealpostmeasurementstate}. We give the precise statement of the theorem below \cite{ChakrabortyPadakandlaSen_22}:
\begin{fact}\label{fact:measurementCompression}
    Given the bipartite quantum state $\rho^{AB}$ and the POVM $\brak{\Lambda_x}_x$ where $x\in \mathcal{X}$, let $\ket{\rho}^{ABR}$ be some purification  of $\rho^{AB}$ and the ideal post measurement state, when the $A$ register of $\rho^{AB}$ is measured using $\Lambda$ is given by:
    \[
    \sum\limits_{x}P_X(x) \ketbra{x}^X\otimes \rho_x^{BR}.
    \]
    Here $P_X$ is the distribution induced by the measurement on the set of symbols $\mathcal{X}$. Suppose we are given integers $K$ and $L$. Then, as long as
     \begin{align*}
        &\log K + \log L \geq H^{'\eps^4}_{\max}(X)+O(\log \frac{1}{\eps}) \\
        &\log L \geq I_{\max}^{\eps^4}(X:RB)+O(\log \frac{1}{\eps}).
    \end{align*}
    there exist POVMs $\Theta^A(1), \Theta^A(2),\ldots \Theta^A(K)$, where each POVM $\Theta_k^A$ has outcomes in the set $[L]\bigcup \brak{\bot}$ ($\bot$ signifying the outcome corresponding o failure), and a function 
    \[
    f : [K]\times [L]\to \mathcal{X}
    \]
    such that
    \[
    \norm{\rho^{XBR}-\sum_x\sum_{k,\ell}Q_{KL}(k,\ell)\cdot \mathtt{1}_{f(k,\ell)=x}\ketbra{x}^X\otimes \sigma_{f(k,\ell)}^{RB}}_1 \leq O(\eps).
    \]
    where $Q_K\overset{\eps^{1/4}}{\approx} \mathbf{Unif}[K]$, $Q_{L|k}$ is the distribution induced on the set $[L]\bigcup \brak{\bot}$ by the POVM $\Theta^A(k)$ and 
    \[
\sigma^{RB}_{f(k,\ell)}\coloneqq \frac{1}{\Tr\left[\Theta_{\ell}^A(k)\ketbra{\rho}^{ARB}\right]} \Tr_A\left[\Theta_{\ell}^A(k)\ketbra{\rho}^{ARB}\right].
\]
\end{fact}
We will make us of a couple of other useful facts about measurement compression, specifically that the distribution $Q_{KL}$ is close to the uniform distribution on $[K]\times [L]$ and that for all $(k,\ell)$ in the support of $Q_{KL}$, the state $\sigma^{RB}_{k,\ell}$ is close to $\rho^{RB}_{f(k,\ell)}$. We state these facts below:

\begin{fact}\label{fact:statescloseQkl}
   For all $k$ in the support of the distribution $Q_K$, and for all $\ell$ in the support of the distribution $Q_{L~|~k}$ (aside from the outcome $\bot$), it holds that
    \[
    \norm{\sigma^{RB}_{f(k,\ell)}-\rho^{RB}_{f(k,\ell)}}_1\leq O(\eps).
    \] 
\end{fact}

\begin{fact}\label{fact:closeQkl}
    Given the setup of Fact \ref{fact:measurementCompression}, it holds that
    \[
    \norm{Q_{KL}-\mathbf{Unif}[K]\times \mathbf{Unif}[L]}_1\leq O(\eps^{1/2}).
    \]
\end{fact}
For a proof of these facts, see the proof of Proposition 4.1 in \cite{ChakrabortyPadakandlaSen_22}.
\subsection{An Improvement in Protocol A: Protocol B}\label{sec:optWithAncilla}
We can now use the measurement compression theorem to design a better protocol for purity distillation than Protocol A. The idea is to make use of the two indices $k$ and $\ell$ that are implicit in the measurement compression theorem. Recall that the index of the POVM to be used in the measurement process is given by $k$. Naturally, Alice and Bob can use this as shared randomness. Although shared randomness is not one of the resources that Alice and Bob are allowed to have for purity distillation, we will soon get rid of it by derandomising. Next, Alice can measure her register $A$ using the POVM $\Theta^A(k)$ indicated by the shared randomness. Since this POVM has only $L$ outcomes, Alice needs to borrow only $L\geq I_{\max}^{\eps}(X:RB)+O(\log \frac{1}{\eps})$ qubits, which is much smaller than $\log \abs{\mathcal{X}}$. By the measurement compression theorem this measurement process produces a state that is close to the ideal post measurement state if Alice had measured with $\Lambda$, after the $(k,\ell)$ indeces have been mapped to appropriate values of $x$. Thus, one would expect, via similar reasoning as that which we used to prove Lemma \ref{lem:averageToWorstcaseHH}, that for \emph{most} setting of $(k,\ell)$, it would hold that:
\[
H_H^{\eps}(\rho^A_{k,\ell})\leq H_H^{\eps^2}(A|X)-\log\eps. 
\]
Alice can then send her $L$ register to Bob via the dephasing channel. Via the same reasoning as above, we expect that for most values of $(k,\ell)$ the following should hold:
\[
H_H^{\eps}(\rho^B_{k,\ell})\leq H_H^{\eps^2}(B|X)-\log\eps. 
\]
Modulo the two assumptions above, this would complete the description of the protocol. Note that the number of qubit states produced by thus protocol would be roughly:
\[
\log \abs{A}-H_H^{\eps^2}(A|X)+\log \abs{B}-H_H^{\eps^2}(B|X)-I_{\max}^{\eps}(X:RB),
\]
where we have suppressed the additive $\log \eps$ terms. One can show that indeed our intuition is correct, as is shown by the following lemma:

\begin{lemma}\label{lem:mainderandhhtwice}
Given the setup of the measurement compression theorem, there exists a subset $\mathcal{S}$ of $[K]\times [L]$ such that
\[
\abs{\mathcal{S}}\geq (1-\eps^{1/8})KL
\]
and for all $(k,\ell)\in \mathcal{S}$ it holds that
it holds that
\begin{align*}
    &H_H^{O(\eps^{1/8})}(RB~|~k,\ell) \leq H_H^{O(\eps)}(RB~|~X)+O(\log\frac{1}{\eps}) \\
    \intertext{and}
    &H_H^{O(\eps^{1/8})}(B~|~k,\ell) \leq H_H^{O(\eps)}(B~|~X)+O(\log \frac{1}{\eps}).
\end{align*}
\end{lemma}

The proof of this lemma is long but does not offer much further insight into the protocol. The reader can find it in Appendix \ref{appendix:derandhh}. To describe our new protocol, we first list the necessary assumptions as required by Fact \ref{fact:measurementCompression}:
\begin{assumption}{\bf Assumptions for Protocol B}
    \begin{enumerate}
    \item Alice and Bob are given a bipartite state $\rho^{AB}$ with purification $\ket{\rho}^{ABR}$, and also a POVM $\Lambda$. The ideal post measurement state when this POVM acts on the register $A$ is given by:
    \[
    \sum\limits_{x} P_X(x)\ketbra{x}\otimes \rho_x^{BR}.
    \]
    \item There exist integers $K$ and $L$ such that
       \begin{align*}
        &\log K + \log L \geq H^{'\eps^4}_{\max}(X)+O(\log \frac{1}{\eps}) \\
        &\log L \geq I_{\max}^{\eps^4}(X:RB)+O(\log \frac{1}{\eps}).
    \end{align*}
    \item Alice possesses the POVMs $\Theta^A(1), \Theta^A(2), \ldots, \Theta^A(K)$ whose existence is implied by Fact \ref{fact:measurementCompression}. Each of these POVMs produces outputs in the set $[L]\bigcup\brak{\bot}$.
    \item We will use the notation $K_AK_B$ to denote a public coin register that is available to both Alice and Bob. We set $\abs{K_A}=\abs{K_B}=K$. The register in which Alice will store the outcome of the measurement will be referred to as $L_A$.
    \item There is a completely dephasing channel from Alice to Bob given by $\mathcal{P}^{L_A\to L_B}$ where $L_A\cong L_B$.
    \item The distribution on the public coin register is given by $Q_K$, as defined in Fact \ref{fact:measurementCompression}.
    \item Given a POVM element $\Theta^A_{\ell}(k)$, we define
    \[
    \ket{\rho_{k,\ell}}^{ABR}\coloneqq \frac{1}{\sqrt{\Tr\left[\Theta^A_{\ell}(k)\rho^{ABR}\right]}}\sqrt{\Theta^A_{\ell}(k)}\ket{\rho}^{ABR}.
    \]
    and its associated marginals of interest accordingly.
 \end{enumerate}
\end{assumption}

We now describe Protocol B in Table \ref{prot:ProtocolB}.

\begin{table}[h]
        \centering
            \begin{tabular}{  l c r }
            \hline \\
            \textbf{Alice} & & \textbf{Bob} \\
            \hline \\
            & Shared public coin &\\
            & in $K_AK_B$ & \\
            \hdashline\\
            Borrow $I_{\max}^{\eps^4}(X:RB)+O(\log \frac{1}{\eps})$ & &\\ pure ancilla qubits in system $L_A$. & & \\
            \hspace{1em} & &\\
            Apply the isometry & &\\
    $\sum_{k} \ketbra{k}^{K_A}\otimes \sum_{\ell} \ket{\ell}^{L_A} \sqrt{\Theta_{\ell}^{A}(k)}$ & & \\
    on the system $A$. & & \\
                \hspace{1em} & &\\
    Define $U_{k,\ell}$ as local optimal distillation & & \\ code for $\rho_{k,\ell}^A$. & & \\
                \hspace{1em} & &\\
    Apply $\sum_{\ell}\ketbra{\ell}^{L_A}\otimes U_{k,\ell}^{A\to A_pA_g}$ on $AL_A$. & &\\
                \hspace{1em} & &\\
    & $\xrightarrow{L_A\to L_B}$ & \\
    & End of Alice's Actions & \\
    \hdashline\\
      & & Do locally optimal protocol on $B$  \\
      & & conditioned on the contents of $K_BL_B$. \\
            \hline
\end{tabular}
        \caption{\bf Protocol B}
        \label{prot:ProtocolB} 
    \end{table}

\begin{proposition}\label{prop:secondProtocol}
Protocol B distils 
\[
\log\abs{A}-H_{H}^{\eps}(A~|~X)+\log \abs{B}-H_H^{\eps}(B~|~X)-I_{\max}^{\eps^4}(X:RB)+O(\log\eps)
\]
number of pure qubits with error $O(\eps^{1/16})$. The protocol also uses $I_{\max}^{\eps^4}(RB:X)+O(\log \frac{1}{\eps})$ amount of classical communication. The entropic quantities above are all computed with respect to the state
\[
\sum_x \ketbra{x}^{X} \otimes \I^{RB}\otimes \Lambda_x\left(\ketbra{\rho}^{ABR}\right) 
\]
where $\ket{\rho}^{ABR}$ is a purification of $\rho^{AB}$.
\end{proposition}

\begin{proof}
    
We will first invoke Lemma \ref{lem:mainderandhhtwice} to note that, for at least $(1-O(\eps^{1/8}))$ fraction of indices $KL$, it holds that
\[
H_H^{O(\eps^{1/8})}(\rho^{RB}_{k,\ell})\leq H_H^{O(\eps)}(RB|X)+O(\log \frac{1}{\eps}).
\]
Note that, for fixed $(k,\ell)$ we recover at least 
\[
\log \abs{A}-H_H^{O(\eps^{1/8})}(\rho^A_{k,\ell})-1
\]
amount of purity. Next, we use the fact that for pure states, such as $\ket{\rho_{k,\ell}}^{ARB}$, Lemma \ref{lem:purehh} implies that:
\[
H_H^{\eps^{1/8}}(\rho^A_{k,\ell})=H_H^{\eps^{1/8}}(\rho^{RB}_{k,\ell})
\]
Thus, using the same arguments as we saw in the proof of Proposition \ref{prop:firstProtocol}, we can show that for all the pairs $(k,\ell)$ which satisfy the above conditions, it holds that:
    \[
\norm{\Tr_{A_
g}\left[U_{k,\ell}\cdot \rho_{k,\ell}^{ARB} \right]-\ketbra{0}^{A_p}\otimes \rho_{k,\ell}^{RB}}_1\leq O(\sqrt{\eps}).
\]
where the system $A_p$ consists of $\log\abs{A}- H_H^{\eps^{1/4}}(RB|X)+O(\log \eps)-O(1)$ qubits. We can then conclude that the global state after Alice sends the system $L_A$ through the dephasing channel satisfies the following condition:
\[
\begin{aligned}
&\left\lVert\sum_{k,\ell} Q_{KL}(k,\ell) \ketbra{k}^K\otimes \ketbra{\ell}^{L_A}\otimes \Tr_{A_g}\left[U_{k,\ell}\cdot \rho_{k,\ell}^{ARB}\right]\right.\\
-&\left.\ketbra{0}^{A_p}\otimes\sum_{k,\ell} Q_{KL}(k,\ell) \ketbra{k}^K\otimes \ketbra{\ell}^{L_A}\otimes\rho_{k,\ell}^{RB}\right\rVert_1\leq O(\eps^{1/8})
\end{aligned}
\]
where we have used the fact that the distribution $Q_{KL}$ is $O(\eps^{1/2})$ to the uniform distribution on $[K]\times [L]$ (see Fact \ref{fact:closeQkl}). Thus, on her side, Alice distils at least
\[
\abs{A_p} \geq\log\abs{A}- H_H^{\eps^{1/4}}(RB|X)+O(\log \eps)-O(1)
\]
amount of purity.

To analyse Bob's actions, we again invoke Lemma \ref{lem:mainderandhhtwice} and recall that, for at least $1-O(\eps^{1/8})$ fraction of indices $KL$, it holds that
\[
H_H^{O(\eps^{1/8})}(\rho^B_{k,\ell}) \leq H_H^{O(\eps)}(B~|~X)+O(\log \frac{1}{\eps}).
\]
This implies that, for most indices $k$ and $\ell$, there exists a local unitary $V^{B\to B_pB_g}_{k,\ell}$ such that
\[
\norm{\Tr_{B_g}\left[V_{k,\ell}\cdot\rho_{k,\ell}^{B}\right]-\ketbra{0}^{B_p}}_1 \leq O(\sqrt{\eps})
\]
where we see that
\[
\abs{B_p}\geq \abs{B}-H_H^{O(\eps)}(B|X)+O(\log \eps)-O(1).
\]
Then, using the fact that the distribution $Q_{KL}$ is $O(\eps^{1/2})$ close to the uniform distribution on $[K]\times [L]$ and the arguments we used for Alice's actions, we see that the following holds:
\[
\begin{aligned}
&\left\lVert\sum_{k,\ell} Q_{KL}(k,\ell) \ketbra{k}^K\otimes \ketbra{\ell}^{L_A}\otimes \Tr_{A_gB_gR}\left[V_{k,\ell}\otimes U_{k,\ell}\cdot \rho_{k,\ell}^{ARB}\right]\right.\\
-&\left.\ketbra{0}^{A_p}\otimes\ketbra{0}^{B_p}\otimes \left(\sum_{k,\ell} Q_{KL}(k,\ell) \ketbra{k}^K\otimes \ketbra{\ell}^{L_A}\right)\right\rVert_1\leq \eps^{1/16}
\end{aligned}
\]
Tracing out all registers but the systems $A_pB_p$ implies the result, where we see that the \emph{net} number of pure qubits that Protocol B distilled is given by:
\[
\log\abs{A}- H_H^{O(\eps)}(RB|X)+\log \abs{B}-H_H^{O(\eps)}(B|X)-I_{\max}^{\eps^4}(X:RB)+O(\log \eps)-O(1).
\]
It is not hard to show that for states of the form
\[
\sum_x \ketbra{x}^{X} \otimes \I^{RB}\otimes \Lambda_x\left(\ketbra{\rho}^{ABR}\right)
\]
it holds that
\[
H_H^{\delta}(RB|X)=H_H^{\delta}(A|X).
\]
Plugging this in into the above expression, the result follows. The claim about the number of bits of classcial communication used by Protocol B follows directly from its specification. This concludes the proof.
\end{proof}
\subsection{Removing the Public Coin from Protocol B: \kd}\label{sec:kd}
In this section we derandomise Protocol B by removing the public coin registers $K_AK_B$ to obtain \kd. We show this in the following lemma:
\begin{lemma}\label{lem:derandK}
    Given the setting of Proposition \ref{prop:secondProtocol}, there exists a subset $\mathcal{T}\subseteq [K]$ of size at least $(1-O(\eps^{1/32}))K$, such that for any $k\in \mathcal{T}$, if Alice runs Protocol B with only the POVM corresponding to this $k$, the resulting protocol, called \textup{\kd}, distils as many pure qubits as Protocol B.
\end{lemma}

\begin{proof}
    Recall that in Protocol B at the end of Bob's actions, the global state satisfied the following property:
    \[
\begin{aligned}
&\left\lVert\sum_{k,\ell} Q_{KL}(k,\ell) \ketbra{k}^K\otimes \ketbra{\ell}^L\otimes \Tr_{A_gB_gR}\left[V_{k,\ell}\otimes U_{k,\ell}\cdot \rho_{k,\ell}^{ARB}\right]\right.\\
-&\left.\ketbra{0}^{A_p}\otimes\ketbra{0}^{B_p}\otimes \left(\sum_{k,\ell} Q_{KL}(k,\ell) \ketbra{k}^K\otimes \ketbra{\ell}^L\right)\right\rVert\leq \eps^{1/16}
\end{aligned}
\]
To derandomise the above protocol, we define
\[
\sigma_k^{A_pB_p}\coloneqq \sum_{\ell}Q_{L~|~k}(\ell~|~k)\Tr_{A_gB_gR}\left[V_{k,\ell}\otimes U_{k,\ell}\cdot \rho_{k,\ell}^{ARB}\right]
\]
Then, using block diagonality, we see that
\[
\sum_k Q_K(k)\norm{\sigma_{k}^{A_pB_p}-\ketbra{0}^{A_p}\otimes \ketbra{0}^{B_p}} \leq \eps^{1/16}.
\]
This immediately proves that there exists a $k$ such that if we run the protocol for only that fixed $k$, Alice and Bob distil the same amount of purity as in the protocol with shared randomness, while making an error at most $\eps^{1/16}$. In fact, since $Q_K$ is $O(\eps^{1/2})$ close to the uniform distribution on $[K]$, this implies that at least $1-O(\eps^{1/32})$ fraction of $k$'s in $[K]$ satisfy this property. This concludes the proof.
\end{proof}

\section{The Protocol With Small Ancilla}\label{sec:optWithoutAncilla}
In the previous section, we proved a lower bound on $\kappa_{\eps}^{\to}(\rho^{AB}, C_{\textup{classical}})$. We did this by first fixing a POVM $\Lambda^A$ on the system $A$, then using the \kd protocol to extract roughly $\log \abs{AB}-H_H^{\eps}(A|X)-H_H^{\eps}(B|X)-I^{\eps^4}_{\max}(X:RB)$ pure qubits. In the process Alice was required to borrow roughly $I_{\max}^{\eps^4}(X:RB)$ many ancilla qubits. The lower bound on $\kappa_{\eps}^{\to}(\rho^{AB}, C_{\textup{classical}})$ was then obtained by taking an infimum over the  POVMs $\Lambda$ over the subset $\mathcal{S}$ (see Theorem \ref{thm:UBgeneral} for specifics).

In this section we show that, given the same fixed POVM $\Lambda$, there exists a protocol which we call \name, which uses far fewer ancilla qubits than \kd, yet manages to distil the pure qubits at the same rate as that of \kd. In fact, we show in Corollary \ref{corol: general compare} in Section \ref{sec:compare} that \name outperforms \kd in terms of the number of qubits borrowed by Alice, as long as $\rho^A$ is not very close to the maximally mixed state. We also show in Corollary \ref{corol:unboundedcommcompare} in Section \ref{sec:compare} that in the case when unbounded classical communication is allowed, \name always outperforms \kd.

We state the main theorem of this section below, which shows the existence of the \name protocol. As in Section \ref{sec:LBdistributed}, we assume that we are given a fixed POVM $\Lambda$ to state and prove our results.

\begin{theorem}{\bf Main Theorem }\label{thm:informal}
    Consider a bipartite state $\rho^{AB}$ shared between two parties Alice and Bob, and a POVM $\brak{\Lambda_x}_x$ where the symbol $x$ belongs to a set of symbols $\mathcal{X}$. Let $\ket{\rho}^{ABR}$  be a purification of $\rho^{AB}$. Consider the control state
    \[
    \rho^{ABRX}\coloneqq \sum\limits_{x}\ketbra{x}^X\otimes \Lambda_x^{A}\left(\ketbra{\rho}^{ABR}\right).
    \]
Let us refer to the number of pure qubits that Alice and Bob can distil as $\textup{Purity}_{\textsc{alice}}$ and $\textup{Purity}_{\textsc{bob}}$. Then, there exists a protocol \name where Alice and Bob are allowed only local unitary operations, and one way classical communication from Alice to Bob such that:

\begin{enumerate}
    \item {\bf Case I:} If $
    I_{\max}^{\eps^4}(RB:X)+H_{H}^{O(\eps)}(RB|X)-O(\log \eps)\leq \log \abs{A}$, they are able to distil the following number of pure qubits:
\begin{align*}
    \textup{Purity}_{\textsc{alice}}&\geq \log \abs{A}-I_{\max}^{\eps^4}(RB:X)+H_{H}^{O(\eps)}(RB|X)+O(\log \eps), \\
    \textup{Purity}_{\textsc{bob}}&\geq \log \abs{B}-H_H^{\eps^2}(B|X)+O(\log \eps).
\end{align*}
while borrowing at most $O(\log \frac{1}{\eps})$ ancilla qubits.

\item {\bf Case II:} If $
    I_{\max}^{\eps^4}(RB:X)+H_{H}^{O(\eps)}(RB|X)-O(\log \eps)> \log \abs{A}$,
they are able to distil 
\begin{align*}
    \textup{Purity}_{\textsc{alice}}&= 0, \\
    \textup{Purity}_{\textsc{bob}}&\geq \log \abs{B}-H_H^{\eps^2}(B|X)-\Delta(RB|X)+O(\log \eps).
\end{align*}
while borrowing at most $
\Delta(RB|X)\coloneqq H_H^{\eps^2}(RB|X)-H_{\min}^{O(\eps)}(RB|X)-O(\log\eps)$
many qubits.
\end{enumerate}
 All entropic quantities above are computed with respect to the control state.
\end{theorem}

\begin{proof}
    The proof is implied by Lemma \ref{lem:unitaryOne} and Lemma \ref{lem:unitaryTwo}.
\end{proof}

Our main task now  is to prove Lemma \ref{lem:unitaryOne} and \ref{lem:unitaryTwo}. To do this, let us start by examining Alice's actions in Protocol C, as described in the last section (see Lemma \ref{lem:derandK}). To recap, Alice and Bob share a public coin register $K$, and based on the contents of $K$, Alice implements a POVM $\Theta^A(k)$ coherently on her system $A$. She stores the outcome of this measurement a system $L_A$ which she creates by borrowing roughly $I_{\max}^{\eps^4}(X:RB)$ qubits, where the entropic quantity is computed with respect to a control state 
\[
\sum\limits_{x}\ketbra{x}^{X}\otimes \Lambda_x^{A}\left(\ketbra{\rho}^{ABR}\right).
\]
Alice then performs a locally optimal distillation protocol on the state $\rho_{k,\ell}^A$ using the unitary $U_{k,\ell}^A$. For most values of $k$ and the measurement outcome $\ell$, the measurement compression theorem then implies that the number of pure qubits that Alice distils is at least $\log \abs{A}-H_H^{\eps^2}(A|X)$ (suppressing the additive $O(\log \eps)$ term).

Later, we derandomised and showed that the public coin register is actually not necessary and for most  settings  ($1-O(\eps^{1/32})$ fraction) of the public coin $k$, the corresponding POVM $\Theta^A(k)$ does as well as the randomised protocol. Alice can then choose any of the $k$ from the set $\mathcal{T}$ (see Lemma \ref{lem:derandK} for the definition of $\mathcal{T}$) and run the protocol suing the fixed POVM $\Theta^A(k)$. 

Our goal in this section will be to implement the action of $\Theta^A(k)$ \emph{in place}, that is, by borrowing little to no ancilla qubits. To do this we require $\Theta^A(k)$ to have the property that for \emph{most} outcomes $\ell\in [L]\bigcup \brak{\bot}$, corresponding to the POVM element $\Theta^A_{\ell}(k)$, it holds that:

   \begin{align*}
    &H_H^{\eps}(\rho^{RB}_{k,\ell}) \leq H_H^{O(\eps)}(RB~|~X)+O(\log\frac{1}{\eps}) \\
    \intertext{and}
    &H_H^{\eps}(\rho^{B}_{k,\ell}) \leq H_H^{O(\eps)}(B~|~X)+O(\log \frac{1}{\eps}),
\end{align*}
 where we define:
 \[
 \rho^{RB}_{k,\ell}\coloneqq \Tr_A\frac{\sqrt{\Theta^A_{\ell}(k)}\cdot \ketbra{\rho}^{ABR}}{\Tr\left[\Theta^A_{\ell}(k)\ketbra{\rho}^{ABR}\right]}
 \]
  for some purification $\ket{\rho}^{ABR}$ of $\rho^{AB}$.  That such a POVM exists is shown below via Lemma \ref{lem:newpovmstart} and Claim \ref{claim:goodk}.
    \begin{lemma}
 \label{lem:newpovmstart}
For the setting of Lemma \ref{lem:mainderandhhtwice}, there exists a subset $\mathcal{T}'\subseteq [K]$ of size at least $(1-\eps^{1/16})K$, such that for all $k\in \mathcal{T}'$, there exists a subset $\nice_{L~|~k}\subseteq [L]$ of size  at least $(1-\eps^{1/16})L$ such that for all $k\in \mathcal{T}'$ and $\ell\in \nice_{L~|~k}$:
 \begin{align*}
    &H_H^{O(\eps^{1/8})}(RB~|~k,\ell) \leq H_H^{O(\eps)}(RB~|~X)+O(\log\frac{1}{\eps}) \\
    \intertext{and}
    &H_H^{O(\eps^{1/8})}(B~|~k,\ell) \leq H_H^{O(\eps)}(B~|~X)+O(\log \frac{1}{\eps}).
\end{align*}.
\end{lemma}

 \begin{claim}\label{claim:goodk}
     Consider the setting of Lemma \ref{lem:derandK}. There exists a $k\in \mathcal{T}$ such that the corresponding POVM $\Theta^A(k)$ satisfies the requirements of Lemma \ref{lem:newpovmstart}.
 \end{claim}
The proofs of Lemma \ref{lem:newpovmstart} and Claim \ref{claim:goodk} can be found in Appendix \ref{appndx:derandlemmaandclaim}. The next idea is that the states $\rho^{AB}_{k,\ell}$, for all $\ell\in \nice_{L~|~k}$, can be perturbed to a nearby state $\wtt{\rho}^{RB}_{k,\ell}$ by throwing away the smallest eigenvalues which sum to $O(\eps^{1/8})$. We can then consider a purification $\ket{\wtt{\rho}_{k,\ell}}^{A_gRB}$ of $\rho_{k,\ell}^{RB}$, where this system $A_g$ requires only $\exp\left(H_H^{O(\eps)}(RB|X)\right)$ dimensions.

The key idea then is to define an embedding of the systems $L_A$ (which holds the measurement outcomes) and the system $A_g$ into a space of dimension roughly $\exp\left(I_{\max}^{\eps^4}(RB:X)+H_H^{O(\eps)}(RB|X)\right)$. We do this by defining an appropriate pure state on the systems $A_pL_AA_gRB$ using the states $\ket{\wtt{\rho}}^{A_gRB}$, and then using Uhlmann's theorem. This Uhlmann isometry gives us Alice's required unitary. Some care is necessary here since not all $\rho_{k,\ell}^{RB}$ have the nice property we need to define the states $\ket{\wtt{\rho}}^{A_gRB}$. A detailed exposition of these ideas can be found in the proof of the lemma below:

\begin{lemma}\label{lem:unitaryOne}
    Suppose that
    \[
    I_{\max}^{\eps^4}(RB:X)+H_{H}^{O(\eps)}(RB|X)-O(\log \eps)\leq \log \abs{A},
    \]
    where all the entropic quantities are computed with respect to the control state:
    \[
    \sum\limits_{x}\ketbra{x}^X\otimes \Lambda_x^A\left(\ketbra{\rho}^{ABR}\right),
    \]
    where $\ket{\rho}^{ABR}$ is a purification of $\rho^{AB}$. Then there exists a unitary operator $U^{A\to A_pA_gL_A}$ and a system $A_g$ such that:
    \begin{align*}
        &\norm{\Tr_{A_gB_g}\left[V\circ\mathcal{P}\circ U\left(\rho^{AB}\right)\right]-\ketbra{0}^{A_p}\otimes \ketbra{0}^{B_p}}_1 \leq \eps^{1/16}, \\
        & \log \abs{A_p}\geq \log \abs{A}-I_{\max}^{\eps^4}(RB:X)+H_{H}^{O(\eps)}(RB|X)+O(\log \eps)-1, \\
        & \log \abs{B_p}\geq \log \abs{B}-H_H^{\eps^2}(B|X)+O(\log \eps).
    \end{align*}
    where $V^{L_BB\to B_pB_g}$ encapsulates Bob's unitary operations.
\end{lemma}

\begin{proof}
    Let us start with the POVM $\Theta^A(k)$ which satisfies the requirements of both Lemma \ref{lem:newpovmstart} and Lemma \ref{lem:derandK}. That such a $k\in [K]$ and $\Theta^A(k)$ exist is shown in Claim \ref{claim:goodk}. 

    For ease of notation, we henceforth omit the $k$ throughout this proof. This means that we will refer to $\Theta^A(k)$ simply as $\Theta^A$ and the set $\nice_{L~|~k}$ (as given by Lemma \ref{lem:newpovmstart}) simply as $\nice_L$.

We define:
\[
\rho_{\ell}^{RB}\coloneqq \frac{\Tr_A\left[\Theta^A_{\ell}\ketbra{\rho}^{ABR}\right]}{P_{\Theta}(\ell)},
\]
for all $\ell\in [L]\bigcup\brak{\bot}$, where
\[
P_{\Theta}(\ell)\coloneqq \Tr\left[\Theta^A_{\ell}\ketbra{\rho}^{ABR}\right].
\]
We will show that the substate $\sum\limits_{\ell\in \nice_L}P_{\Theta}(\ell)\rho_{\ell}^{RB}$ is close to the state $\sum\limits_{\ell\in [L]\bigcup \brak{\bot}}P_{\Theta}(\ell)\rho_{\ell}^{RB}$. To see this, first recall from the construction of $\Theta^A$ in \cite{ChakrabortyPadakandlaSen_22} that, for all $\ell\neq \bot$,
\[
\Theta_{\ell}^A= \frac{1}{1+O(\eps^{1/4})}\cdot \frac{1}{L}\cdot \left(\rho^{-1/2} \sigma_{\ell} \rho^{-1/2}\right)^A
\]
where the operator $\rho$ is equivalent to the marginal $\rho^A$ of the state $\rho^{AB}$ on the system $A$, and $\sigma_{\ell}$ is a state which arises during the construction, but will not be important in the context of this proof. Also recall that by construction,
\[
\Tr\left[\Theta^A_{\bot}\ketbra{\rho}^{ABR}\right]\leq O(\eps).
\]
Then consider the following:
\begin{align*}
    &\norm{\sum\limits_{\ell\in \nice_L}P_{\Theta}(\ell)\rho_{\ell}^{RB}-\sum\limits_{\ell\in [L]\bigcup \brak{\bot}}P_{\Theta}(\ell)\rho_{\ell}^{RB}}_1\\
    =~&\norm{\sum\limits_{\ell\in \nice_L^c\bigcup{\bot}}P_{\Theta}(\ell)\rho_{\ell}^{RB}}_1\\
    \leq ~& \sum\limits_{\ell\in \nice_L^c}\norm{\Tr_A\left[\Theta_{\ell}^A\ketbra{\rho}^{ABR}\right]}_1+\norm{\Tr_A\left[\Theta^A_{\bot}\ketbra{\rho}^{ABR}\right]}_1\\
    \overset{(a)}{=}~&\sum\limits_{\ell\in \nice_L^c}\Tr\left[\Theta_{\ell}^A\rho^A\right]+\Tr\left[\Theta^A_{\bot}\rho^A\right]\\
    = ~& \sum\limits_{\ell\in \nice_L^c} \frac{1}{1+O(\eps^{1/4})}\cdot \frac{1}{L}\cdot \Tr\left[\sigma_{\ell}\right]+\Tr\left[\Theta^A_{\bot}\rho^A\right]\\
    \overset{(b)}{\leq} ~&  \frac{1}{1+O(\eps^{1/4})}\cdot \frac{\abs{\nice_L^c}}{L}+O(\eps)\\
    \overset{(c)}{\leq} ~& \frac{\eps^{1/16}}{1+O(\eps^{1/4})}+O(\eps)\\
    \leq ~& O(\eps^{1/16}),
\end{align*}
where in step $(a)$ we have used the fact that the $1$-norm of a positive semidefinite matrix is equal to its trace and subsequently traced out the systems $RB$, in step $(b)$ we have used the upper bound on $\Tr\left[\Theta^A_{\bot}\ketbra{\rho}^{ABR}\right]$ and also the structure of the POVM element $\Theta_{\ell}^A$ for $\ell\neq \bot$. In step $(c)$ we have used Lemma \ref{lem:newpovmstart} to bound the size of the set $\nice_L^c$.

An immediate consequence of the above calculation is that:
\[
\sum\limits_{\ell\in \nice_L^c\bigcup \brak{\bot}}P_{\Theta}(\ell)\leq O(\eps^{1/16}).
\]
We then define:
\[
\wtt{P}_{\Theta}(\ell)\coloneqq \frac{P_{\Theta}(\ell)}{\sum\limits_{\ell\in \nice_L} P_{\Theta}(\ell)}~~\forall \ell \in \nice_L.
\]
It is then not hard to see via Lemma 2.1 in \cite{Hayden_2008} that:
\[
\norm{\sum\limits_{\ell\in \nice_L}\wtt{P}_{\Theta}(\ell)\rho_{\ell}^{RB}-\sum\limits_{\ell\in [L]\bigcup \brak{\bot}}P_{\Theta}(\ell)\rho_{\ell}^{RB}}_1 \leq O(\eps^{1/16}).
\]
We will now invoke Lemma \ref{lem:newpovmstart} and the definition of the set $\nice_L$ to see that for all $\ell\in \nice_L$,
\[
H_H^{O(\eps^{1/8})}(\rho^{RB}_{\ell})\leq H_H^{O(\eps)}(RB~|~X)+O(\log \frac{1}{\eps}).
\]
Define 
\[
\wtt{\rho_{\ell}}^{RB}\coloneqq \frac{\sqrt{\Pi_{\ell}}\rho_{\ell}^{RB}\sqrt{\Pi_{\ell}}}{\Tr\left[\Pi_{\ell}\rho_{\ell}^{RB}\right]},
\]
where the operator $\Pi_{\ell}^{RB}$ arises in the definition of $H_H^{O(\eps^{1/8})}(\rho^{RB}_{\ell})$. Note that from \cite{Pranab_notes} we know that for any state $\rho_{\ell}^{RB}$, the expression $H_H^{O(\eps^{1/8})}(\rho_{\ell}^{RB})$ is optimised by an operator which commutes with $\rho_{\ell}^{RB}$ and which has all eigenvalues $1$ aside from maybe the smallest eigenvalue. Furthermore, the $0$ eigenvalues of $\Pi_{\ell}$ coincide with the smallest eigenvalues of $\rho_{\ell}^{RB}$ which add up to at most $O(\eps^{1/8})$ ( see \cite{Pranab_notes} for a proof of these properties). These properties imply the following:
\begin{enumerate}
    \item \label{item:1}The states $\wtt{\rho}^{RB}_{\ell}$ for all $\ell\in \nice_L$ are close to $\rho^{RB}_{\ell}$. To be precise:
    \begin{align*}
        &\norm{\wtt{\rho}_{\ell}-\rho_{\ell}}_1\\
        = ~& \norm{\frac{\sqrt{\Pi_{\ell}}\rho_{\ell}\sqrt{\Pi_{\ell}}}{\Tr\left[\Pi_{\ell}\rho_\ell\right]}-\rho_{\ell}}_1 \\
        \overset{(a)}{\leq} ~& O(\eps^{1/16}),
    \end{align*}
    where in step $(a)$ we have used the fact that $\Tr\left[\Pi_{\ell}\rho_{\ell}\right]\geq 1-\eps^{1/8}$ (by definition of $\Pi_{\ell}$ for all $\ell\in \nice_L$) and the Gentle Measurement Lemma.
    \item \label{item:2}The size support of the support of $\wtt{\rho}_{\ell}^{RB}$ is bounded above by $2^{H_H^{O(\eps^{1/8})}(\rho_{\ell}^{RB})}+1$. To see this, note that 
    \[
    \textup{rank}(\Pi_{\ell})\leq \Tr\left[\Pi_{\ell}\right]+1.
    \]
    Since $\Tr\left[\Pi_{\ell}\right]=2^{H_H^{O(\eps^{1/8})}(\rho_{\ell}^{RB})}$, the claim follows.
\end{enumerate}
A hybrid argument then shows that:
\[
\norm{\sum\limits_{\ell\in \nice_L}\wtt{P}_{\Theta}(\ell)\wtt{\rho}_{\ell}^{RB}-\sum\limits_{\ell\in [L]\bigcup \brak{\bot}}P_{\Theta}(\ell)\rho_{\ell}^{RB}}_1 \leq O(\eps^{1/16}).
\]
Finally, note that the action of $\Theta^A$ on is that of a CPTP map with Kraus operators $\ket{\ell}^{L_A}\sqrt{\Theta_{\ell}}^A$, followed by the trace out operation on the system $A$. Therefore, this cannot change the marginal on the system $RB$, which implies that
\[
\sum\limits_{\ell\in [L]\bigcup \brak{\bot}}P_{\Theta}(\ell)\rho_{\ell}^{RB}=\rho^{RB}.
\]
This gives us the following inequality:
\begin{align}\label{eq:Uhlmann}
\norm{\sum\limits_{\ell\in \nice_L}\wtt{P}_{\Theta}(\ell)\wtt{\rho}_{\ell}^{RB}-\rho^{RB}}_1 \leq O(\eps^{1/16}).
\end{align}
We will now define, for all $\ell\in \nice_L$, the purification $\ket{\wtt{\rho}_{\ell}}^{A_gRB}$ of the state $\wtt{\rho}_{\ell}^{RB}$. Here, the system $A_g$ is a system of dimension $2^{H_H^{O(\eps^{1/8})}(\rho_{\ell}^{RB})}+1$, which is sufficient by the arguments presented in Item \ref{item:2}. It is important to point out that we use the same space $A_g$ to purify \emph{all} the states $\wtt{\rho}_{\ell}^{RB}$.

We further define the pure state:
\[
\ket{\wtt{\rho}}^{L_AA_gRB}\coloneqq \sum\limits_{\ell\in \nice_L}\sqrt{\wtt{P}_{\Theta}(\ell)}\ket{\ell}^{L_A}\ket{\wtt{\rho}_{\ell}}^{A_gRB}.
\]
Note that since $\Theta^A$ has at most $\exp\left(I_{\max}^{\eps^4}(X:RB)-O(\log \eps)\right)$ outcomes (we absorb the additive $1$ due to the $\bot$ outcome in the $O(\log \frac{1}{\eps})$ term), the system $L_AA_g$ is of log dimension:
\[
\begin{aligned}
\log \abs{L_AA_g} &= I_{\max}^{\eps^4}(X:RB)+H_H^{O(\eps^{1/8})}(\rho^{RB}_{\ell})+O(\log \frac{1}{\eps})+O(1) \\
&\overset{(a)}{\leq}  I_{\max}^{\eps^4}(X:RB)+H_H^{O(\eps)}(RB~|~X)+O(\log \frac{1}{\eps})\\
&\overset{(b)}{\leq} \log \abs{A},
\end{aligned}
\]
where for the step $(a)$ we have used Lemma \ref{lem:newpovmstart}, and step $(b)$ is by the hypothesis of the lemma. This implies that there exists a system $A_p$ of log dimension at least $\log \abs{A}-I_{\max}^{\eps^4}(X:RB)-H_H^{O(\eps)}(RB~|~X)-O(\log \frac{1}{\eps})$, such that:
\[
A_pL_AA_g\cong A.
\]
We then define the pure state:
\[
\ket{\wdt{\rho}}^{L_AA_pA_gRB}\coloneqq \ket{0}^{A_p}\ket{\wtt{\rho}}^{L_AA_gRB}.
\]
Note that by Equation \ref{eq:Uhlmann}, 
\[
\wdt{\rho}^{RB}\overset{O(\eps^{1/16})}{\approx}\rho^{RB}.
\]
Therefore, by Uhlmann's theorem, there exists a \emph{unitary operator} $U_{\Theta}: A\to L_AA_pA_g$ such that:
\[
\begin{aligned}
&\norm{\ketbra{\wdt{\rho}}^{L_AA_pA_gRB}-U^{A\to L_AA_pA_g}_{\Theta}\cdot \ketbra{\rho}^{ARB}}_1 \\
=~&\norm{\ketbra{0}^{A_p}\otimes \ketbra{\wtt{\rho}}^{L_AA_gRB}-U^{A\to L_AA_pA_g}_{\Theta}\cdot \ketbra{\rho}^{ARB}}_1 \\
\leq~& O(\eps^{1/32})
\end{aligned}
\]
Next, Alice sends the $L_A$ system through the channel $\mathcal{P}^{L_A\to L_B}$. Then, the following holds by the monotonicity of the $1$-norm:
\begin{align*}
    &\left(\mathcal{P}^{L_A\to L_B}\circ \Tr_{A_g}\circ~ U^{A\to L_AA_pA_g}_{\Theta}\right)\cdot \ketbra{\rho}^{ARB}\\
    \overset{O(\eps^{1/32})}{\approx}~& \ketbra{0}^{A_p}\otimes \left(\sum\limits_{\ell\in \nice_L}\wtt{P}_{\Theta}(\ell)\ketbra{\ell}^{L_B}\otimes \wtt{\rho}_{\ell}^{RB}\right).
\end{align*}
Tracing out the system $R$ and by previous arguments, it is then easy to see that:
\begin{align*}
    &\sum\limits_{\ell\in \nice_L}\wtt{P}_{\Theta}(\ell)\ketbra{\ell}^{L_B}\otimes \wtt{\rho}_{\ell}^{B}\\
    \overset{O(\eps^{1/16})}{\approx}~& \sum\limits_{\ell\in \nice_L}\wtt{P}_{\Theta}(\ell)\ketbra{\ell}^{L_B}\otimes \rho_{\ell}^{B}
\end{align*}
Then, invoking Lemma \ref{lem:newpovmstart}, we see that for all $\rho_{\ell}^B$ in the above expression, it holds that:
\[
H_H^{O(\eps^{1/8})}(\rho_{\ell}^B)\leq H_H^{O(\eps)}(B~|~X)+O(\log \frac{1}{\eps}).
\]
Bob can then enact the conditional unitary:
\[
V^{L_BB\to B_pB_g}= \sum\limits_{\ell\in [L_B]} \ketbra{\ell}^{L_B}\otimes V_{\ell}^{B\to B_pB_g},
\]
where $B\cong B_pB_g$. We define the unitary operators $V_{\ell}^{B\to B_pB_g}$ as follows:
\begin{enumerate}
    \item For all $\ell\in \nice_L$, $V_{\ell}^{B\to B_pB_g}$ performs the locally optimal purity distillation protocol for the state $\rho_{\ell}^B$ with error $O(\eps^{1/16})$.
    \item For all $\ell\in [L_B]\setminus \nice_L$, set $V_{\ell}^{B\to B_pB_g}=\I$, where $\I$ denotes the natural isomorphism between the spaces $B$ and $B_pB_g$.
\end{enumerate}
Then it holds that:
\begin{align*}
     &V^{L_BB\to B_pB_g}\cdot \left(\sum\limits_{\ell\in \nice_L}\wtt{P}_{\Theta}(\ell)\ketbra{\ell}^{L_B}\otimes \rho_{\ell}^{B}\right) \\
     \overset{O(\eps^{1/16})}{\approx}~& \ketbra{0}^{B_p}\otimes \left(\sum\limits_{\ell\in \nice_L}\wtt{P}_{\Theta}(\ell)\ketbra{\ell}^{L_B}\otimes \rho'^{B_g}_{\ell}\right),
\end{align*}
where $\rho'^{B_g}_{\ell}$ is the state on the $B_g$ system remnant after the locally optimal protocol has been enacted on $\rho^B_{\ell}$. Note that the system $B_p$ has log dimension at least:
\[
\log \abs{B_p}\geq \log \abs{B}-H_H^{O(\eps)}(B~|~X)-O(\log \frac{1}{\eps}).
\]
The argument is completed by stringing together all of the above inequalities via a hybrid argument and using the monotonicity of the $1$-norm. This concludes the proof.
\end{proof}

We will now deal with the case when 
\[
I_{\max}^{\eps^4}(RB:X)+H_{H}^{O(\eps)}(RB|X)-O(\log \eps)> \log \abs{A}.
\]

It may happen that for some cases $I_{\max}^{\eps^4}(RB:X)+H_{H}^{O(\eps)}(RB|X)-O(\log \eps)$ exceeds $\log \abs{A}$. In that case, Alice cannot distil any pure qubits. Indeed, she has to borrow some qubits to even implement the unitary $U^A_{\Theta}$, which we constructed in Lemma \ref{lem:unitaryOne}. However, she will be able to implement the unitary $U^A_{\Theta}$ following the same recipe that we showed in Lemma \ref{lem:unitaryOne} if she borrows:
\[
I_{\max}^{\eps^4}(RB:X)+H_{H}^{O(\eps)}(RB|X)-O(\log \eps)-\log \abs{A}
\]
qubits. In this case we will use the following bound on $I_{\max}^{\eps^4}(RB:X)$ which was shown in \cite{Berta_RevShannon}:
\[
\begin{aligned}
I_{\max}^{\eps^4}(RB:X)& \leq H_{\max}^{O(\eps^8)}(RB)-H_{\min}^{O(\eps^8)}(RB|X)-O(\log \eps) \\
&\leq H_H^{O(\eps^8)}(A)-H_{\min}^{O(\eps^8)}(RB|X)-O(\log \eps) \\
&\leq \log \abs{A}-H_{\min}^{O(\eps^8)}(RB|X)-O(\log \eps).
\end{aligned}
\]
Therefore, in this case, Alice would have to borrow at most 
\[
\Delta(RB|X)\coloneqq H_H^{O(\eps)}(RB|X)-H_{\min}^{O(\eps^8)}(RB|X)-O(\log\eps)
\]
many qubits. We state this as a lemma below: \

\begin{lemma}\label{lem:unitaryTwo}
    Given the setting of Lemma \ref{lem:unitaryOne}, suppose that
    \[
    I_{\max}^{\eps^4}(RB:X)+H_{H}^{O(\eps)}(RB|X)-O(\log \eps)> \log \abs{A}.
    \]
    In this case, Alice can implement the unitary $U^{A\to L_AA_g}$ defined in Lemma \ref{lem:unitaryOne} by borrowing at most 
    \[
    \Delta(RB|X)\coloneqq H_H^{O(\eps)}(RB|X)-H_{\min}^{O(\eps^8)}(RB|X)-O(\log\eps)
    \]
    many qubits. Note that there is no $A_p$ system for this case since Alice cannot distil any pure qubits by herself. The net purity that Alice and Bob together distil is given by:
    \[
    \log \abs{B}-H_H^{\eps^2}(B|X)-\Delta(RB|X)+O(\log\eps).
    \]
\end{lemma}

Note that in the asymptotic iid limit, the case dealt with in Lemma \ref{lem:unitaryTwo} does not occur. We conclude by formally showing that \name borrows fewer pure qubits than Protocol C in the worst case, as long as $\log \abs{A}-H_H^{O(\eps)}(A)\geq O(\log \frac{1}{\eps})$, i.e., when the state on the system $A$ is even nominally away from maximally mixed:

\section{Comparative Analysis of \name}\label{sec:compare}

In this section we compare the performance of \name to that of \kd, both in the general case of bounded communication and in the case when unbounded classical communication is allowed. We show that in both cases, as long as $\rho^A$ is not too close to the maximally mixed state, \name outperforms \kd. We state our results as Corollaries \ref{corol: general compare} and \ref{corol:unboundedcommcompare} below, the proofs of which follow as corollaries from Theorem \ref{thm:informal}.

\begin{corollary}\label{corol: general compare}
\textup{\name} borrows fewer qubits as compared to \kd as long as $\log \abs{A}-H_H^{O(\eps)}(A)\geq O(\log \frac{1}{\eps})$.
\end{corollary}

\begin{proof}
    Recall that \kd requires Alice to borrow $C_{\textup{borrow}}\coloneqq I_{\max}^{\eps^4}(RB:X)+O(\log \frac{1}{\eps})$ number of pure qubits to function. On the other hand, \name requires Alice to borrow \[D_{\textup{borrow}}\coloneqq \max\brak{0, I_{\max}^{\eps^4}(RB:X)+H_H^{O(\eps)}(RB|X)-\log \abs{A}+O(\log\frac{1}{\eps})}\] qubits. Clearly, 
    \begin{align*}
        &C_{{\textup{borrow}}}-D_{\textup{borrow}}\\
        \geq & \log \abs{A}-H_H^{\eps}(RB~|~X)-O(\log \frac{1}{\eps})\\
        \geq & \log \abs{A} -H_H^{\eps}(RB)-O(\log \frac{1}{\eps})\\
        = &\log \abs{A} -H_H^{\eps}(A)-O(\log \frac{1}{\eps}).
    \end{align*}
    where we have used the data-processing inequality to show that $H_H^{O(\eps)}(RB~|~X)\leq H_H^{O(\eps)}(RB)$ and Lemma \ref{lem:purehh} to show that $H_H^{O(\eps)}(RB)=H_H^{O(\eps)}(A)$. Therefore the corollary holds as long as $\log \abs{A}-H_H^{O(\eps)}(A)\geq O(\log\frac{1}{\eps})$. This concludes the proof.
\end{proof}

\begin{corollary}\label{corol:unboundedcommcompare}
    Given the setup of Theorem \ref{thm:informal}, suppose that the POVM $\brak{\Lambda_x}_x$ has rank-$1$ elements. Then \name guarantees the following:
    \begin{align*}
        \textup{Purity}_{\textsc{alice}}&\geq \log \abs{A}-H_H^{O(\eps^8)}(A)+O(\log \eps)\\
        \textup{Purity}_{\textsc{bob}}&\geq \log \abs{B}-H_H^{\eps^2}(B)+O(\log \eps)
    \end{align*}
    and the number of qubits that Alice is required to borrow to run the protocol is at most $O(\log \frac{1}{\eps})$.
\end{corollary}

\begin{proof}
    First note that the global state is the pure state $\ket{\rho}^{ABR}=\sum\limits_{i}s_i\ket{i}^A\ket{i}^{RB}$. It is given that the POVM $\Lambda$ is rank-$1$, i.e. it is constituted by operators of the form $\brak{\ketbra{\varphi_x}}_x$. Note that each vector $\ket{\varphi_x}$ has $2$-norm at most $1$. This is simply because each operator $\ketbra{\varphi_x}^A\leq \I^A$ (by the definition of a POVM). By definition of the action of the POVM, the post measurement state $\rho^{XBR}$ is given by:
    \begin{align*}
    &\sum\limits_{x}\ketbra{x}^X\otimes \Tr_{A}\left[\left(I^{RB}\otimes \ketbra{\varphi_x}^A\right)\ketbra{\rho}^{ABR}\right]\\
        = ~&\sum\limits_{x}\ketbra{x}^X\otimes \Tr_{A}\left[\left(I^{RB}\otimes \sqrt{\ketbra{\varphi_x}^A}\right)\cdot\ketbra{\rho}^{ABR}\right].
    \end{align*}
    Note that:
    \begin{align*}
        \left(I^{RB}\otimes \sqrt{\ketbra{\varphi_x}^A}\right)\ket{\rho}^{ABR} &= \left(I^{RB}\otimes \sqrt{\ketbra{\varphi_x}^A}\right)\sum_i s_i\ket{i}^A\ket{i}^{BR}\\
        &= \ket{\wtt{\varphi}_x}^A\left(\sum_i s_i \braket{\varphi_x|i} \ket{i}^{BR}\right)
    \end{align*}
    where $\ket{\widetilde{\varphi}_x}^A$ is the normalised version of the vector $\ket{\varphi_x}^A$ and we have used the fact that $\sqrt{\ketbra{\varphi_x}^A}=\ket{\wtt{\varphi}_x}\bra{\varphi_x}^A$. It is easy to see that:
    \begin{align*}
        \norm{\sum_i s_i \braket{\varphi_x|i} \ket{i}^{BR}}_2^2 &~=  \sum_i s_i^2 \abs{\braket{\varphi_x|i}}^2\\
         &~= \Tr\left[\ketbra{\varphi_x}^A\rho^A\right] \\
         &~\coloneqq P_X(x).
    \end{align*}
    Then, defining $  \ket{\psi_x}^{BR}$ to be the normalised version of the vector   $\sum\limits_{i} s_i \braket{\varphi_x|i} \ket{i}^{RB}$ we can rewrite the post measurement state $\rho^{XBR}$ as:
    \[
    \rho^{XBR}= \sum\limits_{x}P_X(x)\ketbra{x}^X\otimes \ketbra{\psi_x}^{BR}.
    \]
    We know from Lemma \ref{lem:hhzeroforconditionalpure} that for states of this form, it holds that $H_H^{\eps}(RB~|~X)\leq 0$. This implies that for this case:
    \begin{align*}
        &~I_{\max}^{\eps^4}(RB~|~X)+H_H^{O(\eps)}(RB~|~X)+O(\log \frac{1}{\eps}) \\
        &\leq  I_{\max}^{\eps^4}(RB~|~X)+O(\log \frac{1}{\eps}) \\
        &\leq H_{\max}^{O(\eps^8)}(A)+O(\log \frac{1}{\eps}).
    \end{align*}
    where the last line follows from Lemma B.17, \cite{CNB23}. Therefore, by Theorem \ref{thm:informal} this implies that in this case Alice needs to borrow at most $O(\log \frac{1}{\eps})$ ancilla qubits for \name to work. This concludes the proof.
\end{proof}

\section{Acknowledgements}
SC would like to acknowledge that the main question addressed in this paper started as a discussion at PS's house in Kolkata over chai, based on questions not addressed in previous work by SC, AN and FB.

AN and FB acknowledges support from MEXT
Quantum Leap Flagship Program (MEXT QLEAP) Grant ~Number JPMXS0120319794. FB also acknowledge support from MEXT-JSPS Grant-in-Aid for Transformative Research
Areas (A) “Extreme Universe”, No. 21H05183, and from JSPS KAKENHI Grants No. 
19H04066 and No. 20K03746. S.C. worked on this project while he was a postdoctoral fellow at the Centre for Quantum Technologies in Singapore. His was was supported by the National Research Foundation, including under NRF RF Award No. NRF-NRFF2013-13 and NRF2021-QEP2-02-P05 and the
Prime Minister's Office, Singapore and the Ministry of Education, Singapore, under the
Research Centres of Excellence program. The work of R.J. is supported by the NRF grant NRF2021-QEP2-02-P05 and the Ministry of Education, Singapore, under the Research Centres of Excellence program. This work was done in part while R.J. was visiting the Technion Israel Institute of Technology, Haifa, Israel, and the Simons Institute for the Theory of Computing, Berkeley, CA,
USA. S.C. would like to acknowledge support from the National Research Foundation, including under NRF RF Award No. NRF-NRFF2013-13 and NRF2021-QEP2-02-P05 and the Prime Minister's Office, Singapore and the Ministry of Education, Singapore, under the Research Centres of Excellence program. P.S. would like to acknowledge
support of the Department of Atomic Energy, Government of India, under project no. 12-R\&D-TFR-5.01-0500, for carrying out this research work.

\bibliography{ref}
\bibliographystyle{plain}

\appendices{

\section{Proof of Lemma \ref{lem:mainderandhhtwice}}\label{appendix:derandhh}
Before we go on to the main proof, we will prove another lemma which will be useful in the main proof of Lemma \ref{lem:mainderandhhtwice}.
\begin{lemma}\label{lem:condhhderand}
Consider the set of classical symbols $\mathcal{X}$ and let $X$ be a register which holds symbols from this set. Let $B$ be a quantum register. Suppose that we are given two classical quantum states $\rho^{XB}$ and $\sigma^{XB}$ as follows:
\begin{align*}
    &\rho^{XB}\coloneqq \sum\limits_{x} P_X(x)\ketbra{x}^X\otimes \rho_x^B \\
    &\sigma^{XB}\coloneqq \sum\limits_{x} Q_X(x)\ketbra{x}^X\otimes \sigma_x^B
\end{align*}
with the promise that for all $x\in \mathcal{X}$ it  holds that $\norm{\sigma_x^B-\rho_x^B}_1\leq \eps$ and $\norm{P_X-Q_X}_1\leq \eps$. Next, let $K$ be an integer and suppose that we are given a deterministic function $f:[K]\to \mathcal{X}$. Suppose there exists a distribution $Q_K$ on $[K]$ such that the following holds:
\begin{align*}
    &\norm{Q_K-\mathbf{Unif}[K]}_1\leq \delta\\
     \sum\limits_{k: f(k)=x}& Q_K(k)= Q_X(x)~~\forall x\in \mathcal{X}.
\end{align*}
Let us also define, for all $k\in [K]$, the states $\sigma_k^B\coloneqq \sigma_{f(k)}^B$. Then, at least $1-\eps^{1/8}-\delta$ fraction of $k\in [K]$ satisfy the condition that:
\[
2^{H_H^{\eps^{1/8}}(\sigma^B_k)}\leq \frac{2^{H_H^{\eps}(B~|~X)_{\rho}}}{\eps}.
\]
\end{lemma}

\begin{proof}
Let $\Pi_{\textsc{opt}}^{XB}$ be the optimising operator in the definition of $H_H^{\eps}(B~|~X)$, where we can assume without loss of generality that $\Pi_{\textsc{opt}}$ is of the form 
\[
\Pi_{\textsc{opt}}= \sum_{x}\ketbra{x}^X\otimes \Pi_x^B
\]
where each $\Pi_x^B$ satisfies the condition
\[
0^B \leq \Pi_x^B \leq \I^B.
\]
We claim that $\abs{\Tr\left[\Pi_{\textsc{opt}}(\rho^{XB}-\sigma^{XB})\right]}\leq 2\eps$. To see this, note that:
\begin{align*}
    &\abs{\Tr\left[\Pi_{\textsc{opt}}(\rho^{XB}-\sigma^{XB})\right]}\\
    \leq  & \norm{\rho^{XB}-\sigma^{XB}}_1\\
    \leq & 2\eps,
\end{align*}
where the first inequality is by the definition of the $1$-norm and the last inequality can be proved using a standard hybrid argument. This immediately implies that:
\begin{align*}
    &\sum\limits_{x}Q_X(x)\Tr\left[\Pi_x\sigma_x\right]\\& =\Tr\left[\Pi_{\textsc{opt}}\sigma^{XB}\right] \\&\geq     \Tr\left[\Pi_{\textsc{opt}}\rho^{XB}\right]-2\eps  \\
    &\geq 1-3\eps.
\end{align*}
Markov's inequality then implies that there exists a set $\good_1\subseteq \mathcal{X}$ such that $\Pr\limits_{Q_X}\left[\good_1\right]\geq 1-\sqrt{3\eps}$ and for all $x\in \good_1$ it holds that:
\[
\Tr\left[\Pi_x\sigma_x^B\right]\geq 1-\sqrt{3\eps}.    
\]
Since $Q_X$ and $P_X$ are close in the $1$-norm, this implies that $\Pr\limits_{P_X}\left[\good_1\right]\geq 1-\sqrt{4\eps}$. Next, note that by definition, 
\[
\sum\limits_{x}P_X(x)\Tr\left[\Pi_x\right]=2^{-H_H^{\eps}(B~|~X)}.
\]
Again using Markov's inequality, we see that there exists a set $\good_2\subseteq \mathcal{X}$ such that $\Pr\limits_{P_X}\left[\good_2\right]\geq 1-\eps$ and for all $x\in \good_2$ it holds that:
\[
\Tr\left[\Pi_x\right] \leq \frac{2^{H_H^{\eps}}(B~|~X)}{\eps}.
\]
Therefore, we have identified a set $\good_X\coloneqq \good_1\bigcap \good_2$ of probability at least $1-\sqrt{5\eps}$ (under $P_X$) such that for all $x\in \good_X$:
\begin{align*}
    & \Tr\left[\Pi_x\sigma_x^B\right]\geq 1-\sqrt{3\eps}\\
    & \Tr\left[\Pi_x\right] \leq \frac{2^{H_H^{\eps}}(B~|~X)}{\eps}.
\end{align*}
Now, let us define the subset $\good_K\subseteq [K]$ as follows:
\[
\good_K\coloneqq \brak{k~\big|~f(k)=x, x\in \good_X}.
\]
We then define the operator:
\[
\Pi^{\prime KB} \coloneqq \sum\limits_{k\in \good_{K}}\ketbra{k}^{K}\otimes \Pi_{f(k)}^B.
\]
Then observe that :
\begin{align*}
    \Tr\left[\Pi^{'KB}\sigma^{KB}\right] &= \sum\limits_{k\in \good_K}Q_K(k)\Tr\left[\Pi_{f(k)}\sigma_k^B\right] \\
    &= \sum\limits_{x\in \good_X}\sum\limits_{k : f(k)=x} Q_K(k)\Tr\left[\Pi_{f(k)}\sigma_k^B\right] \\
    &= \sum\limits_{x\in \good_X}Q_X(x)\Tr\left[\Pi_{x}\sigma_x^B\right]\\
    & \geq (1-\sqrt{3\eps})\cdot \Pr\limits_{Q_X}\left[\good_X\right] \\
    & \geq 1-\eps^{1/4},
\end{align*}
where the last inequality uses the fact that the probabilities of any set under the distributions $P_X$ and $Q_X$ can differ by at most $\eps$. Again, using Markov's inequality we infer that there exists a subset $\nice_K\subseteq \good_K$ such that $\Pr\limits_{Q_K}\left[\nice_K\right]\geq 1-\eps^{1/8}$, and for all $k\in \nice_K$ it holds that:
\[
\Tr\left[\Pi_{f(k)}\sigma_{k}^B\right] \geq 1-\eps^{1/8}.
\]
This implies that for all $k\in \nice_K$, $\Pi_k^B\coloneqq \Pi_{f(k)}^B$ is a candidate for optimising the expression $2^{H_H^{\eps^{1/8}}(\sigma_k^B)}$. This implies that, for all $k\in \nice_K$:
\begin{align*}
    2^{H_H^{\eps^{1/8}}(\sigma_k^B)} & \leq \Tr\left[\Pi_k^B\right]\\
    &= \Tr\left[\Pi_{x}^B\right]~~\textup{  where } x=f(k)\\
& \leq \frac{2^{H_H^{\eps}(B|X)}}{\eps}~~\textup{  since } x\in \good_X.
\end{align*}
Also note that since $Q_K$ and $\mathbf{Unif}[K]$ are close by $\delta$, it holds that $\Pr\limits_{\mathbf{Unif}[K]}\left[\nice_K\right]\geq 1-\eps^{1/8}-\delta$.  This concludes the proof.
\end{proof}

\subsection*{Proof of Lemma \ref{lem:mainderandhhtwice}}
\begin{proof}
It is not hard to see that the states
\[
\sum\limits_{k,\ell}Q_{KL}(k,\ell)\ketbra{k,\ell}^{KL}\otimes \sigma_{f(k,\ell)}^{RB} \tag{Bob and Ref}\label{eq:BobandRef}
\]
and
\[
\sum\limits_{k,\ell}Q_{KL}(k,\ell)\ketbra{k,\ell}^{KL}\otimes \sigma_{f(k,\ell)}^{B} \tag{Bob}\label{eq:Bob}
\]
\emph{both} satisfy the requirements of Lemma \ref{lem:condhhderand}. For the state in \ref{eq:BobandRef}, we think of the correspondence $B\gets RB$ and $K\gets KL$ with respect to the registers $KB$ in Lemma \ref{lem:condhhderand}. Similarly, for the state in \ref{eq:Bob} the correspondence is $B\gets B$ and $K\gets KL$. Then note that the closeness of $\sigma^{RB}_{f(k,\ell)}$ and $\rho^{RB}_{f(k,\ell)}$ (for $(k,\ell)\in \textup{supp}(Q_{KL})$, implied by Fact \ref{fact:statescloseQkl}) implies the closeness of the marginals $\sigma^{B}_{f(k,\ell)}$ and $\rho^{B}_{f(k,\ell)}$. Also, since $Q_X(x)\coloneqq \sum\limits_{k,\ell : f(k,\ell)=x} Q_{KL}(k,\ell)$, it holds via Fact \ref{fact:measurementCompression} that $\norm{P_X-Q_X}_1 \leq O(\eps)$. Additionally, it also holds that the distribution $Q_{KL}$ is close to the uniform distribution on $[KL]$ by $O(\eps^{1/2})$, as implied by Fact \ref{fact:closeQkl}.

A subtle issue is that the expression in Fact \ref{fact:measurementCompression} the distribution $Q_{L~|~k}$ is supported on $[L]\bigcup \brak{\bot}$ for all $k$. However, from \cite{ChakrabortyPadakandlaSen_22} we know that the mass on this element is at most $O(\eps)$ (for our choice of parameters) and can thus be removed  from the $1$-norm expression with a penalty of at most $O(\eps)$. This allows us to run the argument above for only $\ell\neq \bot$, and the valifity of the closeness of the state $\sigma_{f(k,\ell)}^{RB}$ and $\rho^{RB}_{f(k,\ell)}$ holds.

We will first instantiate parameters $(\eps, \delta)$ in Lemma \ref{lem:condhhderand}, then use Lemma \ref{lem:condhhderand} twice. To that end, set $\delta \gets O(\eps^{1/2})$ and $\eps\gets O(\eps)$. Then, we first apply Lemma \ref{lem:condhhderand} to the state in \ref{eq:BobandRef} to see that there exists a subset
\[
\mathcal{S}_1 \subseteq [K]\times [L]
\]
with the property that
\[
\abs{\mathcal{S}_1} \geq (1-O(\eps^{1/8}))\cdot KL
\]
and for all $(k,\ell)\in \mathcal{S}_1$, it holds that
\[
H_H^{O(\eps^{1/8})}(RB~|~k,\ell) \leq H_H^{O(\eps)}(RB~|~X)+O(\log\frac{1}{\eps}).
\]
Similarly, applying Lemma \ref{lem:condhhderand} to Equation \ref{eq:Bob}, we see that there exists a set $\mathcal{S}_2\subseteq [K]\times [L]$ such that
\[
\abs{\mathcal{S}_2} \geq (1-O(\eps^{1/8}))KL
\]
and for all $(k,\ell)\in \mathcal{S}_2$ it holds that
\[
H_H^{O(\eps^{1/8})}(B~|~k,\ell) \leq H_H^{O(\eps)}(B~|~X)+O(\log \frac{1}{\eps}).
\]
It holds then that for all $(k,\ell)$ in the set
\[
\mathcal{S}\coloneqq \mathcal{S}_1\bigcap \mathcal{S}_2,
\]
where
\[
\abs{\mathcal{S}}\geq (1-O(\eps^{1/8}))KL,
\]
it holds that
\begin{align*}
    &H_H^{O(\eps^{1/8})}(RB~|~k,\ell) \leq H_H^{O(\eps)}(RB~|~X)+O(\log\frac{1}{\eps})\\
    \intertext{and}
    &H_H^{O(\eps^{1/8})}(B~|~k,\ell) \leq H_H^{O(\eps)}(B~|~X)+O(\log \frac{1}{\eps}).
\end{align*}
This concludes the proof.
\end{proof}

\section{Proofs of Lemma \ref{lem:hhzeroforconditionalpure} and \ref{lem:hhswitchforconditionalpure}}\label{appendix:hhlemmas}

\subsection*{Proof of Lemma \ref{lem:hhzeroforconditionalpure}}

\begin{proof}
    Suppose $\Pi^{XB}_{\opt}$ is the optimising operator that for the quantity $H_H^{\eps}(B~|~X)$. Without loss of generality we can assume that $\Pi_{\opt}$ is of the following form:
    \[
    \Pi^{XB}_{\opt}=\sum\limits_{x}\ketbra{x}^X\otimes \Pi_x^B.
    \]
    Then by definition $\Pi_{\opt}$ satisfies the following optimisation problem:
    \begin{align*}
        \min\limits_{\brak{\Pi_x}_x~:~ 0\leq \Pi_x\leq \I} & \sum\limits_{x}P_X(x)\Tr\left[\Pi_x\right] \\
        & \sum\limits_{x}P_X(x)\Tr\left[\Pi_x\ketbra{v_x}\right]\geq 1-\eps.
    \end{align*}
    Let us define for every $x\in \mathcal{X}$ an operator:
    \[
    \Pi'_{x}\coloneqq \Tr\left[\Pi_x\ketbra{v_x}\right] \ketbra{v_x},
    \]
    and
    \[
\Pi^{' XB}=\coloneqq \sum\limits_{x}\ketbra{x}^X\otimes \Pi^{' B}_{x}.
    \]
    It is clear that $\Tr\left[\Pi^{' XB}\rho^{XB}\right]\geq 1-\eps$ and that 
    $\sum\limits_{x}P_X(x)\Tr\left[\Pi^{'XB}\right]\leq 1$. This implies that $\Pi^{' XB}$ is a candidate optimiser for $H_H^{\eps}(B~|~X)$ and thus:
    \[
    H_H^{\eps}(B~|~X)\leq 0.
    \]
    This concludes the proof.
\end{proof}

\subsection*{Proof of Lemma \ref{lem:hhswitchforconditionalpure}}

\begin{proof}
    We can assume without loss of generality that $\abs{A}\leq \abs{B}$ and that the optimising operator $\Pi_{\opt}^{XA}$ is of the form:
    \[
    \Pi_{\opt}^{XA}=\sum\limits_{x}\ketbra{x}^X\otimes \Pi_x^A.
    \]
    Let us fix $x\in \mathcal{X}$. Let the corresponding $\ket{v_x}^{AB}$ have the following Schmidt decomposition:
    \[
    \ket{v_x}^{AB}=\sum\limits_{i}\lambda_i \ket{a_i}^A\ket{b_i}^B.
    \]
    Let $V^{A\to B}_x$ be an isometry defined by $\ket{a_i}^A\to \ket{b_i}^B$ for all $i$. Then, consider the following:
    \begin{align*}
    &\Tr\left[\Pi_{\opt}^{XB}\rho^{XAB}\right]\\
        =&\sum\limits_{x}P_X(x)\Tr\left[\Pi_x^A\ketbra{v_x}^{AB}\right]\\
        = & \sum\limits_{x}P_X(x)\sum\limits_{i}\lambda_i^2\Tr\left[\Pi_x^A\ketbra{a_i}^A\right]
\\ = & \sum\limits_{x}P_X(x)\sum\limits_{i}\lambda_i^2\Tr\left[\Pi_x^AV^{\dagger }_x\ketbra{b_i}^BV_x\right]\\
= & \sum\limits_{x}P_X(x)\sum\limits_{i}\lambda_i^2\Tr\left[\left(V_x\Pi_x^AV^{\dagger }_x\right)^B\ketbra{b_i}^B\right] \\
= & \sum\limits_{x}P_X(x)\Tr\left[\left(V_x\Pi_x^AV^{\dagger }_x\right)^B\ketbra{v_x}^{AB}\right].
\end{align*}
Therefore, we can define an operator $\widetilde{\Pi}^{XB}=\sum\limits_{x}P_X(x)\ketbra{x}\otimes \widetilde{\Pi}_x^B$ where for each $x\in \mathcal{X}$, $\widetilde{\Pi}_x^B\coloneqq \left(V_x\Pi_x^A V_x^{\dagger}\right)^B$, such that:
\begin{align*}
    &\Tr\left[\widetilde{\Pi}^{XB}\rho^{XB}\rho^{XB}\right]\\
    =& \Tr\left[\Pi^{XA}_{\opt}\rho^{XA}\rho^{XA}\right]\\\
    \geq & 1-\eps.
\end{align*}
Therefore, $\widetilde{\Pi}^{XB}$ is a candidate optimiser for the quantity $H_H^{\eps}(B~|~X)$, which implies that:
\[
H_H^{\eps}(B~|~X)\leq H_H^{\eps}(A~|~X).
\]
A similar analysis shows that:
\[
H_H^{\eps}(A~|~X)\leq H_H^{\eps}(B~|~X).
\]
This concludes the proof.
\end{proof}

\section{Proofs of Lemma \ref{lem:newpovmstart} and Claim \ref{claim:goodk}}\label{appndx:derandlemmaandclaim}
\subsection{Proof of Lemma \ref{lem:newpovmstart}}
\begin{proof}
 From Lemma \ref{lem:mainderandhhtwice}, we know that the entropic inequalities in the statement of the lemma hold for at least $(1-O(\eps^{1/8}))$ fraction of all index pairs $(k,\ell)$. Define  $\one_{k,\ell}$ as the indicator that the entropic inequalities hold for the fixed index pair $(k,\ell)$. Then, 
 \[
 \sum\limits_{k,\ell}\frac{1}{KL}\one_{k,\ell} \geq (1-O(\eps^{1/8})).
 \]
 Define
 \[
 \texttt{prob}_{k}\coloneqq \sum_{k,\ell}\frac{1}{L} \one_{k,\ell}
 \]
 Then, it holds by Markov's inequality that for $(1-\eps^{1/16})$ fraction of $k$'s, 
 \[
 \textsc{prob}_k \geq 1-\eps^{1/16}.
 \]
We define the set $\mathcal{T}'$ to be that set of $k$'s where the above condition holds. Let $k\in \mathcal{T}'$. Then by the fact that $\texttt{prob}_{k}$ is an average of indicator functions, we can conclude that, for at least $1-\eps^{1/16}$ fraction of $\ell$'s in $[L]$, it holds that
 \[
 \one_{k,\ell}=1 .
 \]
 We define the set where the above condition holds to be $\nice_{L~|~k}$. This concludes the proof.
 \end{proof}

 \subsection{Proof of Claim \ref{claim:goodk}}

 \begin{proof}
     Recall that Corollary \ref{lem:newpovmstart} shows the existence of the set $\mathcal{T}'\subseteq [K]$ of size at least $(1-\eps^{1/16})K$. On the other hand, Lemma \ref{lem:derandK} shows that the set $\mathcal{T}$ is of size at least $(1-O(\eps^{1/32}))K$. This implies that $\abs{\mathcal{T}\bigcap \mathcal{T}'}> 0$. This concludes the proof.
 \end{proof}

\end{document}